\RequirePackage{etoolbox}
\newtoggle{journalVersion}

 \togglefalse{journalVersion} 

\newcommand{\inJournal}[1]{\iftoggle{journalVersion}{#1}{}} 
\newcommand{\inArXiv}[1]{\iftoggle{journalVersion}{}{#1}}  

\inArXiv{\documentclass[11pt]{article}}
\inJournal{
	\documentclass[prodmode,acmtalg]{acmsmall}
}
\usepackage{latexsym}
\usepackage{amssymb}
\usepackage{amsmath}
\inArXiv{
	\usepackage{graphicx}
	\usepackage{amsthm}
	\usepackage[paper=letterpaper,margin=1in]{geometry}
	\usepackage{paralist}
	\usepackage{comment}
}
\usepackage{forloop}
\usepackage[ruled,vlined,linesnumbered]{algorithm2e}
\usepackage{nicefrac}

\inJournal{
	
	\SetAlFnt{\small}
	\SetAlCapFnt{\small}
	\SetAlCapNameFnt{\small}
	\SetAlCapHSkip{0pt}
	\IncMargin{-\parindent}
}

\newcommand{\enumerateEnv}{\inArXiv{compactenum}\inJournal{enumerate}}
\newcommand{\itemizeEnv}{\inArXiv{compactitem}\inJournal{itemize}}

\inArXiv
{
	\newtheorem{theorem}{Theorem}[section]
	\newtheorem{lemma}[theorem]{Lemma}
	\newtheorem{corollary}[theorem]{Corollary}

	\newtheorem{reduction}{Reduction}
}
\newtheorem{observation}[theorem]{Observation}
\inJournal
{
	\newtheorem{reduction}[theorem]{Reduction}
}

\makeatletter
\inArXiv{
	\newtheorem*{rep@theorem}{\rep@title}
	\newcommand{\newreptheorem}[2]{%
	\newenvironment{rep#1}[1]{%
		\def\rep@title{#2 \ref{##1}}%
		\begin{rep@theorem}}%
		{\end{rep@theorem}}}
}
\inJournal{
	\newcommand{\newreptheorem}[2]{%
		\newenvironment{rep#1}[1]{%
		\expandafter\renewcommand\csname the#1\endcsname{\ref{##1}}%
		\begin{#1}}%
		{\end{#1}%
		\addtocounter{#2}{-1}}}
}
\makeatother

\inArXiv{
	\newreptheorem{theorem}{Theorem}
	\newreptheorem{lemma}{Lemma}
}
\inJournal{
	\newreptheorem{theorem}{theorem}
	\newreptheorem{lemma}{theorem}
}

\newcommand{\defcal}[1]{\expandafter\newcommand\csname c#1\endcsname{{\mathcal{#1}}}}
\newcommand{\defbb}[1]{\expandafter\newcommand\csname b#1\endcsname{{\mathbb{#1}}}}
\newcounter{calBbCounter}
\forLoop{1}{26}{calBbCounter}{
    \edef\letter{\Alph{calBbCounter}}
    \expandafter\defcal\letter
		\expandafter\defbb\letter
}

\newcommand{\eps}{\varepsilon}
\newcommand{\cOPT}{{\overline{OPT}}}
\newcommand{\characteristic}{\mathbf{1}}
\newcommand{\RSet}{\mathbf{\cR}}
\newcommand{\SW}{\textup{\texttt{SW}}}
\newcommand{\SWfull}{\texttt{Submodular Welfare}}
\newcommand{\opt}{\mathtt{opt}}
\inArXiv{
	\newcommand{\ie}{{\it i.e.}}
	\newcommand{\eg}{{\it e.g.}}
	\newcommand{\citeWithName}[2]{#1~\cite{#2}}
}
\inJournal{
	\newcommand{\ie}{{i.e.}}
	\newcommand{\eg}{{e.g.}}
	\newcommand{\qedhere}{\qed}
	\newcommand{\citeWithName}[2]{\citeN{#2}}
}

\newcommand{\titleString}{Maximizing Symmetric Submodular Functions}
\title{\inArXiv{\textbf{\titleString}}\inJournal{\titleString}\footnote{This work has been supported in part by ERC Starting Grant 335288-OptApprox. An extended abstract of this work appeared in ESA 2015.}}
\inArXiv{
\author{Moran Feldman\thanks{Department of Mathematics and Computer Science, Open University of Israel.
Email:
\texttt{moranfe@openu.ac.il}.}}
}
\inJournal{
\author{MORAN FELDMAN \affil{EPFL}}
}

\inArXiv{
	\begin{document}
	\maketitle
}
\begin{abstract}
Symmetric submodular functions are an important family of submodular functions capturing many interesting cases including cut functions of graphs and hypergraphs. Maximization of such functions subject to various constraints receives little attention by current research, unlike similar minimization problems which have been widely studied. In this work, we identify a few submodular maximization problems for which one can get a better approximation for symmetric objectives than the state of the art approximation for general submodular functions.

We first consider the problem of maximizing a non-negative symmetric submodular function $f\colon 2^\cN \to \bR^+$ subject to a down-monotone solvable polytope $\cP \subseteq [0, 1]^\cN$. For this problem we describe an algorithm producing a fractional solution of value at least $0.432 \cdot f(OPT)$, where $OPT$ is the optimal \emph{integral} solution. Our second result considers the problem $\max \{f(S) : |S| = k\}$ for a non-negative symmetric submodular function $f\colon 2^\cN \to \bR^+$. For this problem, we give an approximation ratio that depends on the value $k / |\cN|$ and is always at least $0.432$. Our method can also be applied to non-negative \emph{non-symmetric} submodular functions, in which case it produces $\nicefrac{1}{e} - o(1)$ approximation, improving over the best known result for this problem. For unconstrained maximization of a non-negative symmetric submodular function we describe a \emph{deterministic linear-time} $\nicefrac{1}{2}$-approximation algorithm. Finally, we give a $[1 - (1 - 1/k)^{k - 1}]$-approximation algorithm for {\SWfull} with $k$ players having identical non-negative submodular utility functions, and show that this is the best possible approximation ratio for the problem.
\end{abstract}

\inArXiv{\noindent \textbf{Keywords:} Symmetric submodular functions, cardinality constraint, matroid constraint}

\inArXiv{\excludecomment{CCSXML}}

\inJournal
{
	\markboth{Moran Feldman}{\titleString}


	\terms{Algorithms, Theory}
	\acmformat{Moran Feldman. 2016. \titleString.}
	\begin{document}
	\begin{bottomstuff}
	Author’s addresses: (Current address) M. Feldman, Department of Mathematics and Computer Science, The Open University of Israel, 1 University Road, POB 808, Raanana 4353701, Israel.
	\end{bottomstuff}
	\maketitle
}
\inArXiv
{
\thispagestyle{empty}
\newpage
\setcounter{page}{1}
}

\newcommand{\ProofGeneralProperties}
{
\begin{proof}
The first part of the lemma holds since:
\[
	\bar{F}(x)
	=
	F(\characteristic_\cN - x)
	=
	\bE[f(\RSet(\characteristic_\cN - x))]
	=
	\bE[f(\cN \setminus \RSet(x))]
	=
	\bE[\bar{f}(\RSet(x))]
	\enspace.
\]
Using the above observation, the second part of the lemma follows since, for a symmetric $f$,
\[
	F(x)
	=
	\bE[f(\RSet(x))]
	=
	\bE[\bar{f}(\RSet(x))]
	=
	\bar{F}(x)
	\enspace.
\]
Finally, for the third part of the lemma, let us assume $\RSet(x), \RSet(y)$ and $\RSet(z)$ are chosen using the following process: for every element $u \in \cN$ an independent and uniformly random threshold $t_u \in [0, 1]$ is selected. Then, $u$ is added to $\RSet(x), \RSet(y)$ or $\RSet(z)$ if $t_u \leq x_u$, $t_u \leq y_u$ or $t_u \leq z_u$, respectively. Observe that this process indeed results in sets $\RSet(x), \RSet(y)$ and $\RSet(z)$ having the right distributions. Moreover, $\RSet(z) \subseteq \RSet(y) \subseteq \RSet(x)$. Thus,
\begin{align*}
	F(x) - F(y)
	={} &
	\bE[f(\RSet(x)) - f(\RSet(y))]
	\leq
	\bE[f(\RSet(x) \setminus \RSet(z)) - f(\RSet(y) \setminus \RSet(z))]\\
	={} &
	\bE[f(\RSet(x - z))] - \bE[f(\RSet(y - z))]
	=
	F(x - z) - F(y - z)
	\enspace.
	\qedhere
\end{align*}
\end{proof}
}

\newcommand{\ProofSingeltonOK}
{
\begin{proof}
An element $u \in \cN$ such that $\characteristic_u \not \in \cP$ cannot appear in any integral solution. Thus, removing all such elements from $\cN$ results in a new instance with the same value of $\max \{F(x) : x \in \cP \cap \{0,1\}^\cN\}$. Moreover, such a removal can only increase $d(\cP)$, and thus, the guarantee of Theorem~\ref{thm:polytope_symmetric} for the new polytope must be as strong as for the original polytope.
\end{proof}
}

\newcommand{\LemmaStepApproximation}
{
\begin{lemma} \label{lem:step_approximation}
For every time $0 \leq t < T$, $\sum_{u \in \cN} (1 - y_u(t)) \cdot I_u(t) \cdot \partial_u F(y(t)) \geq F(y(t) \vee \characteristic_{OPT}) - F(y(t))$.
\end{lemma}
\begin{proof}
Let us calculate the weight of $OPT$ according to the weight function $w(t)$.
\[
    w(t) \cdot \characteristic_{OPT}
    =
    \sum_{u \in OPT} w_u(t)
    =
    \sum_{u \in OPT} \left[F(y(t) \vee \characteristic_u) - F(y(t))\right]
		\geq
    F(y(t) \vee \characteristic_{OPT}) - F(y(t))
		\enspace,
\]
where the inequality follows from submodularity. Since $\characteristic_{OPT} \in \cP$, we get:
\[
    w(t) \cdot I(t)
		\geq
		w(t) \cdot \characteristic_{OPT}
		\geq
		F(y(t) \vee \characteristic_{OPT}) - F(y(t))
		\enspace.
\]
Hence,
\begin{align*}
    \sum_{u \in \cN} (1 - y_u(t)) \cdot I_u(t) \cdot \partial_u F(y(t))
    &=
    \sum_{u \in \cN} I_u(t) \cdot [F(y(t) \vee \characteristic_u) - F(y(t))]
    =
    I(t) \cdot w(t)\\
    &\geq
    F(y(t) \vee \characteristic_{OPT}) - F(y(t)) \enspace. \qedhere
\end{align*}
\end{proof}
}

\newcommand{\ProofStepImprovement}
{
\begin{proof}
Notice that $\max_{u \in \cN} f(\{u\}) \leq f(OPT)$ by Reduction~\ref{re:singleton_ok}. Hence, by combining Lemmata~\ref{lem:greedy_step_increase_bound_general} and~\ref{lem:step_approximation}, we get:
\[
	F(y(t + \delta)) - F(y(t))
	\geq
	\delta \cdot [F(y(t) \vee \characteristic_{OPT}) - F(y(t))] - O(n^3\delta^2) \cdot f(OPT)
	\enspace,
\]
where $y(t + \delta)$ represents its value before the loop starting on Line~\ref{line:second_scan} of Algorithm~\ref{alg:MeasuredContinuousGreedyKeepDown}. The corollary follows by noticing that the above loop can only increase the value of $F(y(t + \delta))$.
\end{proof}
}

\newcommand{\LemmaRemovingOnlyDecrease}
{
\begin{lemma} \label{lem:removing_only_decrease}
$F(x) \leq F(y(t))$ for every $0 \leq t \leq T$ and vector $x \in [0, 1]^\cN$ such that $x \leq y(t)$.
\end{lemma}
\begin{proof}
First observe that the lemma is trivial for $t = 0$ since $y(0) = \characteristic_\varnothing$. Hence, we assume in the rest of the proof $t > 0$.

Let $u_1, u_2, \ldots, u_n$ be the order in which the algorithm scans the elements in the loop starting on Line~\ref{line:second_scan}. Let $y^i(t)$ be the vector $y(t)$ immediately after the iteration of this loop corresponding to $u_i$. Notice that $y(t) = y^n(t)$. Then,
\begin{align*}
	F(y(t))
	={} &
	F(x) + \sum_{i = 1}^n \int_0^1 (y_{u_i}(t) - x_{u_i}) \cdot \partial_{u_i} F((1 - z)x + z \cdot y(t)) dz\\
	\geq{} &
	F(x) + \sum_{i = 1}^n  \int_0^1 (y_{u_i}(t) - x_{u_i}) \cdot \partial_{u_i} F(y^i(t)) dz
	\enspace,
\end{align*}
where the equality follows from the chain rule and the inequality follows from submodularity and the observation that $x \leq y(t) \leq y^i(t)$. The algorithm guarantees that for every $u_i \in \cN$, either $\partial_{u_i} F(y^i(t)) \geq 0$ or $x_{u_i} = y_{u_i}(t) = 0$. Notice that, in both cases, $(y_{u_i}(t) - x_{u_i}) \cdot \partial_{u_i} F(y^i(t)) \geq 0$.
\end{proof}
}

\newcommand{\ProofGoodSolution}
{
\begin{proof}
Combining Lemmata~\ref{lem:union_bound_symmetric} and~\ref{lem:removing_only_decrease} implies $F(y(t) \vee \characteristic_{OPT}) \geq f(OPT) - F(y(t))$. The corollary now follows by plugging this inequality into Corollary~\ref{cor:step_improvment}.
\end{proof}
}

\newcommand{\LemmataGH}
{
Let $g(t)$ be defined as follows: $g(0) = 0$ and $g(t + \delta) = g(t) + \delta [f(OPT) - 2 \cdot g(t)]$. The next lemma shows that a lower bound on $g(t)$ also gives a lower bound on $F(y(t))$

\begin{lemma} \label{lem:g_bound2}
For every $0 \leq t \leq T$, $g(t) \leq F(y(t)) + O(n^3 \delta) \cdot t \cdot f(OPT)$.
\end{lemma}
\begin{proof}
Let $c$ be the constant hiding behind the big $O$ notation in Corollary~\ref{cor:good_solution}. We prove by induction on $t$ that $g(t) \leq F(y(t)) + c n^3\delta^2 t \cdot f(OPT)$. For $t = 0$, $g(0) = 0 \leq F(y(0))$. Assume now that the claim holds for some $t$, and let us prove it for $t + \delta$. Corollary~\ref{cor:good_solution} gives:
\begin{align*}
    g(t + \delta)
    &=
    g(t) + \delta [f(OPT) - 2 \cdot g(t)]
    =
    (1 - 2\delta) g(t) + \delta \cdot f(OPT)\\
    &\leq
    (1 - 2\delta) [F(y(t)) + cn^3\delta t \cdot f(OPT)] + \delta \cdot f(OPT)\\
    &=
    F(y(t)) + \delta [f(OPT) - 2 \cdot F(y(t))] + c(1 - 2\delta) n^3\delta t \cdot f(OPT)\\
    &\leq
    F(y(t + \delta)) + cn^3\delta^2 \cdot f(OPT) + c(1 - 2\delta)n^3\delta t \cdot f(OPT)\\
    &\leq
    F(y(t + \delta)) + cn^3\delta(t + \delta) \cdot f(OPT)
		\enspace.
		\qedhere
\end{align*}
\end{proof}

The function $g$ is given by a recursive formula, thus, evaluating it is not immediate. Instead, we show that the function $h(t) = \nicefrac[]{1}{2} \cdot [1 - e^{-2t}] \cdot f(OPT)$ lower bounds $g$.
\begin{lemma} \label{lem:h_bound2}
For every $0 \leq t \leq T$, $g(t) \geq h(t)$.
\end{lemma}
\begin{proof}
The proof is by induction on $t$. For $t = 0$, $g(0) = 0 = \nicefrac[]{1}{2} \cdot [1 - e^{-2\cdot 0}] \cdot f(OPT) = h(0)$. Assume now that the lemma holds for some $t$, and let us prove it holds for $t + \delta$.
\begin{align*}
    h(t + \delta)
    &=
    h(t) + \int_t^{t + \delta} h'(\tau) d\tau
    =
    h(t) + f(OPT) \cdot \int_t^{t + \delta} e^{-2\tau} d\tau\\
    &\leq
    h(t) + f(OPT) \cdot \delta e^{-2t}
    =
    (1 - 2\delta) h(t) + \delta \cdot f(OPT)\\
    &\leq
    (1 - 2\delta) g(t) + \delta \cdot f(OPT)
    =
    g(t) + \delta [f(OPT) - 2 \cdot g(t)]
    =
    g(t + \delta) \enspace. \qedhere
\end{align*}
\end{proof}
}

\newcommand{\ProofReductionBase}
{
\begin{proof}
The reduction follows immediately from the proof of Corollary~5.3 in~\cite{LMNS10}. The idea is that if this is not the case, then let $\bar{k} = n - k$. It can be verified that $2\bar{k} \leq n$, that the problem $\max\{\bar{f}(S) : |S| = \bar{k}\}$ is equivalent to the original problem and that $\bar{f}$ is a a non-negative symmetric submodular function if and only if $f$ has these properties (in fact, if $f$ is symmetric then $f = \bar{f}$).
\end{proof}
}

\newcommand{\ProofYProperties}
{
\begin{proof}
We first prove the first part of the lemma by induction on $t$. For $t = 0$ the claim is trivial. Assume the claim holds for time $t$, and let us prove it for time $t + \delta$. By definition, for every element $u \in \cN$, either:
\[
	y^1_u(t + \delta) = 0
	\qquad
	\text{or}
	\qquad
	y^1_u(t + \delta) = y^1_u(t) + \delta I^1_u(t) \cdot (1 - y^1_u(t)) \leq y^1_u(t) + (1 - y^1_u(t)) = 1
	\enspace.
\]
Similarly, we also get either:
\[
	y^2_u(t + \delta) = 1
	\qquad
	\text{or}
	\qquad
	y^2_u(t + \delta) = y^2_u(t) - \delta I^2_u(t) \cdot y^2_u(t) \geq 0
	\enspace.
\]

To prove the second part of the lemma, let $y^3(t) = \characteristic_\cN - y^2(t)$ for every time $t$. It is easy to see that for every time $0 \leq t < T$ and element $u \in \cN$:
\[
	y^3_u(t + \delta) \leq y^3_u(t) + \delta I^2_u(t) \cdot (1 - y^3_u(t))
	\enspace.
\]
Think of $(y^1_u(t), y^3_u(t))$ as a point in the polytope $y^1_u + y^3_u \leq 1$. The density of this polytope is $1/2$, and thus, by Lemma~\ref{lem:feasibility},
\[
	y^1_u(T) + y^3_u(T) \leq 1
	\Rightarrow
	y^1_u(T) \leq y^2_u(T)
\]
as long as $T \leq -\nicefrac{n}{k} \cdot \ln(1 - k/n +n^{-4}) \leq -2\ln(1/2 +n^{-4})$.

To prove the third part of the lemma, notice that the densities of both polytopes $\sum_{u \in \cN} x_u \leq k$ and $\sum_{u \in \cN} x_u \leq n - k$ are at least $k/n$. Thus, by Lemma~\ref{lem:feasibility}, $|y^1(T)| \leq k$ and $|y^3(T)| \leq n - k$ as long as $T \leq -\nicefrac{n}{k} \cdot \ln(1 - k/n +n^{-4})$. The third part of the lemma follows from these observations since $|y^2(T)| = n - |y^3(T)|$.
\end{proof}
}

\newcommand{\ProofFeasiblity}
{
\begin{proof}
Consider first the case $|y^2(T)| = |y^1(T)|$. Since $y^1(T) \leq k \leq y^2(T)$, we must have $|y^1(T)| = k$. Hence, $y = y^1(T)$ is indeed feasible in this case.

Consider now the case $|y^1(T)| \neq |y^2(T)|$. In this case, the vector $y$ is defined by:
\[
	y^1(T) \cdot \frac{|y^2(T)| - k}{|y^2(T)| - |y^1(T)|} + y^2(T) \cdot \frac{k - |y^1(T)|}{|y^2(T)| - |y^1(T)|}
	\enspace.
\]
Observe that $y$ is a convex combination of $y^1(T)$ and $y^2(T)$, and thus, it is a vector in $[0, 1]^\cN$. Moreover,
\[
	|y|
	=
	|y^1(T)| \cdot \frac{|y^2(T)| - k}{|y^2(T)| - |y^1(T)|} + |y^2(T)| \cdot \frac{k - |y^1(T)|}{|y^2(T)| - |y^1(T)|}
	=
	k
	\enspace.
	\qedhere
\]
\end{proof}
}

\newcommand{\ProofConcave}
{
\begin{proof}
The non-negativity of $r$ follows immediately from the non-negativity of $F$. Thus, it only remains to prove that $r$ is concave. Let $\hat{r}(x) = y^1(T) + x(y^2(T) - y^1(T))$, \ie, $r(t) = F(\hat{r}(t))$. By the chain rule,
\[
	\frac{d r(x)}{d x}
	=
	\frac{d F(\hat{r}(x))}{d x}
	=
	\sum_{u \in \cN} \left[ \frac{\partial \hat{r}_u(x)}{\partial x} \cdot \partial_u F(\hat{r}(x))\right]
	=
	\sum_{u \in \cN} \left[ (y^2_u(T) - y^1_u(T)) \cdot \partial_u F(\hat{r}(x))\right]
	\enspace.
\]
By Lemma~\ref{lem:y_properties}, $y^2_u(T) - y^1_u(T)$ is a non-negative constant for every $u \in \cN$. On the other hand, by submodularity, $\partial_u F(\hat{r}(x))$ is a non-increasing function of $x$ since $\hat{r}(x)$ is a linear non-decreasing function (coordinate-wise). Hence, $\frac{d r(x)}{d x}$ is a non-increasing function of $x$.
\end{proof}
}

\newcommand{\ProofStepApproximationTwoSided}
{
\begin{proof}
Let us calculate the weights of $OPT$ and $\cOPT$ according to the weight functions $w^1(t)$ and $w^2(t)$, respectively.
\[
    w^1(t) \cdot \characteristic_{OPT}
    =
    \sum_{u \in OPT} w^1_u(t)
    =
    \sum_{u \in OPT} \left[F(y^1(t) \vee \characteristic_u) - F(y^1(t))\right]
		\geq
    F(y^1(t) \vee \characteristic_{OPT}) - F(y^1(t)) \enspace,
\]
and
\[
    w^2(t) \cdot \characteristic_{\cOPT}
    =
    \sum_{u \in \cOPT} w^2_u(t)
    =
    \sum_{u \in \cOPT} \left[F(y^2(t) \wedge \characteristic_{\cN - u}) - F(y^2(t))\right]
		\geq
    F(y^2(t) \wedge \characteristic_{OPT}) - F(y^2(t)) \enspace,
\]
where the inequalities follow from submodularity. Since $|OPT| = k$,
\begin{align*}
	&
	\min\{I^1(t) \cdot w^1(t) + 2 \cdot F(y^1(t)), I^t(t) \cdot w^2(t) + 2 \cdot F(y^2(t)) \}\\
	\geq{} &
	\min\{\characteristic_{OPT} \cdot w^1(t) + 2 \cdot F(y^1(t)), \characteristic_{\cOPT} \cdot w^2(t) + 2 \cdot F(y^2(t)) \}\\
	\geq{} &
	\min\{F(y^1(t) \vee \characteristic_{OPT}) + F(y^1(t)), F(y^2(t) \wedge \characteristic_{OPT}) + F(y^2(t)) \}
	=
	\Delta(t)
	\enspace.
\end{align*}
Hence,
\begin{align*}
    \sum_{u \in \cN} (1 - y^1_u(t)) \cdot I^1_u(t) \cdot \partial_u F(y^1(t))
    &=
    \sum_{u \in \cN} I^1_u(t) \cdot [F(y^1(t) \vee \characteristic_u) - F(y^1(t))]
    =
    I^1(t) \cdot w^1(t)\\
    &\geq
    \Delta(t) - 2 \cdot F(y^1(t)) \enspace,
\end{align*}
and
\begin{align*}
    -\sum_{u \in \cN} y^2_u(t) \cdot I^2_u(t) \cdot \partial_u F(y^2(t))
    &=
    \sum_{u \in \cN} I^2_u(t) \cdot [F(y^2(t) \wedge \characteristic_{\cN - u}) - F(y^2(t))]
    =
    I^2(t) \cdot w^2(t)\\
    &\geq
    \Delta(t) - 2 \cdot F(y^2(t)) \enspace.
		\qedhere
\end{align*}
\end{proof}
}

\newcommand{\ProofOptNotBad}
{
\begin{proof}
Let $u^*$ be the element of $\cN$ for which $f(\{u^*\}) = \max_{u \in \cN} f(\{u\})$, and let $A,B \subseteq \cN - u^*$ be two disjoint subsets of size $k - 1$ (there are such sets since $|\cN - u^*| = n - 1 \geq 2(k - 1)$). Then,
\[
	f(OPT)
	\geq
	\frac{f(A + u^*) + f(B + u^*)}{2}
	\geq
	\frac{f(\{u^*\})}{2}
	\enspace.
	\qedhere
\]
\end{proof}
}

\newcommand{\ProofDeltaBound}
{
\begin{proof}
It can be easily verified that Lemma~\ref{lem:removing_only_decrease} applies here (for $y^1(t)$), \ie, for every vector $x \leq y^1(t)$, $F(x) \leq F(y^1(t))$. Combining this observation with Lemma~\ref{lem:union_bound_symmetric} gives:
\begin{equation} \label{eqn:bound_positive}
	F(y^1(t) \vee \characteristic_{OPT})
	\geq
	f(OPT) - F(y^1(t))
	\enspace.
\end{equation}

Using an analogous proof to the one of Lemma~\ref{lem:removing_only_decrease}, it can be shown that for every vector $x \geq y^2(t)$, $F(x) \leq F(y^2(t))$. This implies that $\bar{F}(x) \leq \bar{F}(\characteristic_\cN - y^2(t))$ for every $x \leq \characteristic_\cN - y^2(t)$. Combing this observation with Lemma~\ref{lem:union_bound_symmetric} gives:
\begin{equation} \label{eqn:bound_negative}
	F(y^2(t) \wedge \characteristic_{OPT})
	=
	\bar{F}((\characteristic_\cN - y^2(t)) \vee \characteristic_{\cOPT})
	\geq
	\bar{f}(\cOPT) - \bar{F}(\characteristic_\cN - y^2(t))
	=
	f(OPT) - F(y^2(t))
	\enspace.
\end{equation}

The lemma now follows by plugging inequalities~\eqref{eqn:bound_positive} and~\eqref{eqn:bound_negative} into the definition of $\Delta(t)$.
\end{proof}
}
\section{Introduction} \label{sec:introduction}

The study of combinatorial problems with submodular objective functions has recently attracted much attention, and is motivated by the principle of economy of scale, prevalent in real world applications. Submodular functions are also commonly used as utility functions in economics and algorithmic game theory. Symmetric submodular functions are an important family of submodular functions capturing, for example, the mutual information function and cut functions of graphs and hypergraphs.

Minimization of symmetric submodular functions subject to various constrains and approximating such functions by other functions received a lot of attention~\cite{DDSSS12,D09,GS13,N10,Q98}. However, maximization of symmetric submodular functions was the subject of only limited research, despite an extensive body of works dealing with maximization of general non-negative submodular functions (see, \eg,~\cite{BFNS15,CCPV11,CVZ14,LSV10,V13}). In fact, we are only aware of two papers dealing with maximization of symmetric submodular functions.
First, \citeWithName{Feige et al.}{FMV11} show an $\nicefrac[]{1}{2}$-approximation algorithm for the problem of maximizing a symmetric submodular function subject to no constraint (which is the best possible). This result was later complemented by an algorithm achieving the same approximation ratio for general submodular functions~\cite{BFNS15}. Second, \citeWithName{Lee et al.}{LMNS10} show a $\nicefrac[]{1}{3}$-approximation algorithm for maximizing a symmetric submodular function subject to a general matroid base constraint.

In this work, we identify a few submodular maximization problems for which one can get a better approximation for symmetric objectives than the state of the art approximation for general submodular functions. Our first result is an improved algorithm for maximizing a non-negative symmetric submodular function\footnote{A set function $f\colon 2^\cN \to \bR^+$ is \emph{symmetric} if $f(S) = f(\cN \setminus S)$ for every set $S \in \cN$, and \emph{submodular} if $f(A) + f(B) \geq f(A \cup B) + f(A \cap B)$ for every pair of sets $A, B \subseteq \cN$.} $f\colon 2^\cN \to \bR^+$ subject to a down-monotone solvable polytope\footnote{A polytope $\cP \subseteq [0, 1]^\cN$ is \emph{solvable} if one can optimize linear functions subject to it, and \emph{down-monotone} if for every two vectors $x, y \in [0, 1]^\cN$, $x \leq y$ and $y \in \cP$ imply $x \in \cP$.} $\cP \subseteq [0, 1]^\cN$. More formally, given a set function $f\colon 2^\cN \to \bR$, its \emph{multilinear} extension is the function $F\colon [0, 1]^\cN \to \bR$ defined by $F(x) = \bE[f(\RSet(x))]$, where $\RSet(x)$ is a random set containing every element $u \in \cN$ with probability $x_u$, independently. Our result is an approximation algorithm for the problem $\max \{F(x) : x \in \cP\}$ whose approximation ratio is about: $\nicefrac[]{1}{2} \cdot [1 - (1 - d(\cP)^{2/d(\cP)})]$, where $d(\cP)$ is the density\footnote{Consider a representation of $\cP$ using $m$ inequality constraints, and let $\sum_{u \in \cN} a_{i, u} x_u \leq b_i$ denote the $i^{th}$ inequality constraint. By Section~3.A of~\cite{F13}, we may assume all the coefficients are non-negative and each constraint has at least one non-free non-zero coefficient. The {\em density} $d(\cP)$ of $\cP$ is defined as the maximum value of $\min_{1 \leq i \leq m} \frac{b_i}{\sum_{u \in \cN} a_{i,u}}$ for any such representation.} of $\cP$. In the following theorem, and throughout the paper, we use $n$ to denote $|\cN|$.

\begin{theorem} \label{thm:polytope_symmetric}
Given a non-negative symmetric submodular function $f\colon 2^\cN \to \bR^+$, a down-monotone solvable polytope $\cP \subseteq [0,1]^\cN$ and a constant $T \geq 0$, there exists an efficient algorithm that finds a point $x \in [0,1]^\cN$ such that $F(x) \geq \nicefrac[]{1}{2} \cdot [1 - e^{-2T} - o(1)] \cdot \max \{F(x) : x \in \cP \cap \{0,1\}^\cN\}$.
Additionally,
\begin{\enumerateEnv}[(a)]
    \item $x / T \in \cP$. \label{cond:feasibility_1}
    \item Let $T_{\cP} = -\ln(1 - d(\cP) + n^{-4})/d(\cP)$. Then, $T \leq T_{\cP}$ implies $x \in \cP$. \label{cond:feasibility_2}
\end{\enumerateEnv}
\end{theorem}

Theorem~\ref{thm:polytope_symmetric} improves over the result of~\cite{FNS11}, who gave an approximation ratio of $e^{-1} - o(1)$ for the case of general submodular functions. More specifically, Theorem~\ref{thm:polytope_symmetric} provides an approximation ratio of at least $\nicefrac[]{1}{2} \cdot [1 - e^{-2}] - o(1) \geq 0.432$ for an arbitrary down-monotone solvable polytope since $T$ can always be set to be at least $1$. For many polytopes the fractional solution produced by Theorem~\ref{thm:polytope_symmetric} can be rounded using known rounding methods (see, \eg, pipage rounding~\cite{CCPV11}, swap rounding~\cite{CVZ11a} and contention resolution schemes~\cite{CVZ14}). For example, matroid polytopes allow rounding without any loss in the approximation ratio. Moreover, due to property~\eqref{cond:feasibility_1} of Theorem~\ref{thm:polytope_symmetric}, the combination of our algorithm with the contention resolution schemes rounding described by~\cite{CVZ14} produces better approximation ratios than might be expected by a black box combination (see~\cite{FNS11} for details).

Our next result considers the problem $\max \{f(S) : |S| = k\}$ for a non-negative symmetric submodular function $f\colon 2^\cN \to \bR^+$. For this problem we prove the following theorem.

\begin{theorem} \label{thm:uniform_base_symmetric_approximation}
There exists an efficient algorithm that given a non-negative symmetric submodular function $f\colon 2^\cN \to \bR^+$ and an integer cardinality parameter $1 \leq k \leq n/2$, achieves an approximation of $\nicefrac[]{1}{2}[1 - (1 - k/n)^{2n/k}] - o(1)$ for the problem: $\max \{f(S) : |S| = k\}$. If $k > n / 2$, then the same result holds with the cardinality parameter replaced by $n - k$.
\end{theorem}

Notice that Theorem~\ref{thm:uniform_base_symmetric_approximation} achieves for the problem $\max \{f(S) : |S| = k\}$ the same approximation ratio achieved by Theorem~\ref{thm:polytope_symmetric} for the problem $\max \{f(S) : |S| \leq k\}$ (as long as $k \leq n/2$). Using the same technique we get a result also for the more well-studied case of general (non-symmetric) submodular functions.

\begin{theorem} \label{thm:uniform_base_general_approximation}
There exists an efficient algorithm that given a non-negative submodular function $f\colon 2^\cN \to \bR^+$ and an integer cardinality parameter $1 \leq k \leq n$, achieves an approximation of $e^{-1} - o(1)$ for the problem: $\max \{f(S) : |S| = k\}$.
\end{theorem}

Theorems~\ref{thm:uniform_base_symmetric_approximation} and~\ref{thm:uniform_base_general_approximation} improve over results achieved by~\cite{BFNS14} when $k/n \leq 0.204$ and $k/n \leq 0.093$, respectively. Most practical applications of maximizing a submodular function subject to a cardinality constraint use instances having relatively small $k/n$ ratios, and thus, can benefit from our improvements (see~\cite{BFNS14} for a list of such applications). We complement Theorem~\ref{thm:uniform_base_symmetric_approximation} by showing that one cannot get an approximation ratio better than $\nicefrac[]{1}{2}$ for any ratio $k/n$.

\begin{theorem} \label{thm:hardness_uniform_symmetric}
Consider the problems $\max\{f(S) : |S| = p/q \cdot n\}$ and $\max\{f(S) : |S| \leq p/q \cdot n\}$ where $p < q$ are positive constant integers and $f$ is a non-negative symmetric submodular function $f\colon 2^\cN \to \bR^+$ obeying $n / q \in \bZ$. Then, every algorithm with an approximation ratio of $\nicefrac[]{1}{2} + \eps$ for one of the above problems (for any constant $\eps > 0$) uses an exponential number of value oracle queries.\footnote{See Section~\ref{sec:preliminaries} for the definition of value oracles.}
\end{theorem}

The result of Theorem~\ref{thm:hardness_uniform_symmetric} follows quite easily from the symmetry gap framework of~\cite{V13} and is known for the case of general submodular functions as well as for some pairs of $p$ and $q$ (\eg, the case $\nicefrac[]{p}{q} = \nicefrac[]{1}{2}$ follows immediately from the work of~\cite{V13}). We give the theorem here mainly for completeness, and defer its proof to Appendix~\ref{sec:hardness}.

We also consider the unconstrained submodular maximization problem (\ie,\inJournal{ the problem} $\max \{f(S) : S \subseteq \cN\}$). For symmetric submodular functions, \citeWithName{Feige et al.}{FMV11} give for this problem a simple linear-time randomized algorithm and a slower deterministic local search, both achieving an optimal approximation ratio of $\nicefrac[]{1}{2}$ (up to a low order error term in the case of the local search). We show that for such functions there exists also a \emph{deterministic linear-time} $\nicefrac[]{1}{2}$-approximation algorithm.

\begin{theorem} \label{thm:deterministic_symmetric_unconstrained}
There exists a deterministic linear-time $\nicefrac[]{1}{2}$-approximation algorithm for the problem $\max \{f(S) : S \subseteq \cN\}$, where $f\colon 2^\cN \to \bR^+$ is a non-negative symmetric submodular function.
\end{theorem}

Theorem~\ref{thm:deterministic_symmetric_unconstrained} improves over the time complexity of the local search algorithm of~\cite{FMV11} and also avoids the low order error term. It is interesting to note that a deterministic algorithm with the same approximation ratio (but a worse time complexity) for the case of general submodular functions was only very recently presented by \citeWithName{Buchbinder and Feldman}{BF16}.

Our final result considers a variant of the {\SWfull} problem (\SW). An instance of {\SW} consists of $m$ players $p_1, p_2, \dotsc, p_m$ and $n$ items $\cN$. Each player $p_i$ is associated with a non-negative submodular utility function $u_p\colon 2^\cN \to \bR^+$. The objective is to find a partition $S_1, S_2, \dotsc S_m$ of the items maximizing $\sum_{i = 1}^m u_p(S_i)$. We consider the case of identical utility functions, \ie, the utility function $u_p$ is identical for all players. This problem is interesting for two reasons. First, it generalizes $\max \{f(S) : S \subseteq \cN\}$ for symmetric submodular functions.\footnote{When $f$ is symmetric, the problem $\max \{f(S) : S \subseteq \cN\}$ is equivalent to {\SW} with two players having $f$ as their common utility function.} Second, it is related to the Submodular Multiway Partition problem considered by~\cite{CE11b,CE11,EVW13,ZNI05}.

\begin{theorem} \label{thm:identical_submodular_welfare}
There exists a linear-time $[1 - (1 - 1/k)^{k - 1}]$-approximation algorithm for {\SW} with $k$ players having identical non-negative submodular utility functions. Moreover, any algorithm for this problem whose approximation ratio is $[1 - (1 - 1/k)^{k - 1}] + \eps$ (for some constant $\eps > 0$) must use an exponential number of value oracle queries.
\end{theorem}

Theorem~\ref{thm:identical_submodular_welfare} improves over a result of \citeWithName{Iwata et al.}{ITY16}, who give \inArXiv{a $\nicefrac{1}{2}$-approximation}\inJournal{an approximation of $\nicefrac{1}{2}$} for {\SW} with identical non-negative \emph{symmetric} utility functions. Interestingly, Theorem~\ref{thm:identical_submodular_welfare} also shows that {\SW} with identical utility functions is a rare example of a submodular maximization problem with a non-monotone\footnote{A submodular function $f\colon 2^\cN \to \bR^+$ is \emph{monotone} if $f(A) \leq f(B)$ for every two sets $A \subseteq B \subseteq \cN$.} objective having an approximation ratio strictly better than $\nicefrac[]{1}{2}$ (for $k > 2$). On the other hand, the hardness result of Theorem~\ref{thm:identical_submodular_welfare} complements a result of \citeWithName{Khot et al.}{KLMM08} who showed that, even when the utility functions have a succinct representation (and thus, can be evaluated directly instead of being accessed by a value oracle), no polynomial time algorithm can obtain a better than $(1 - 1/e)$-approximation for {\SW} with identical \emph{monotone} utility functions unless $P = NP$.

\subsection{Our Techniques}

Some of our results are based on variants of the measured continuous greedy algorithm of~\cite{FNS11}. We modify the measured continuous greedy in two main ways.

\inArXiv{\begin{itemize}}\inJournal{\begin{longitem}}
\item The analysis of~\cite{FNS11} relies on the observation that $F(\characteristic_{OPT} \vee x) \geq [1 - \max_{u \in \cN} x_u] \cdot f(OPT)$ for an arbitrary vector $x \in [0, 1]^\cN$.\footnote{For every set $S \subseteq \cN$, we use $\characteristic_{S}$ to denote the characteristic vector of $S$. Given two vectors $x, y \in [0, 1]^\cN$, we use $x \vee y$ to denote the coordinate-wise maximum of $x$ and $y$. In other words, for every $u \in \cN$, $(x \vee y)_u = \max\{x_u, y_u\}$. Similarly, $x \wedge y$ denotes the coordinate-wise minimum of $x$ and $y$.} To get better results for symmetric functions we use an alternative lower bound on $F(\characteristic_{OPT} \vee x)$ given by Lemma~\ref{lem:union_bound_symmetric}.

\begin{lemma} \label{lem:union_bound_symmetric}
Given a non-negative symmetric submodular function $f\colon 2^\cN \to \bR^+$, a set $S \subseteq \cN$ and a vector $x \in [0, 1]$ obeying $F(y) \leq F(x)$ for every $\{y \in [0, 1]^\cN : y \leq x\}$, then $F(\characteristic_S \vee x) \geq f(S) - F(x)$.
\end{lemma}

Using the bound given by Lemma~\ref{lem:union_bound_symmetric} in the analysis requires a slight modification of the measured continuous greedy algorithm to guarantee that its solution always obeys the requirements of the lemma. We defer the proof of Lemma~\ref{lem:union_bound_symmetric} to Section~\ref{sec:preliminaries}.

\item The measured continuous greedy algorithm can handle only constraints specified by a down-monotone polytope. Thus, it cannot handle problems of the form $\max \{f(S) : |S| = k\}$. To bypass this difficulty, we use two instances of the measured continuous greedy algorithm applied to the problems $\max \{f(S) : |S| \leq k\}$ and $\max \{{f}(\cN \setminus S) : |S| \leq n - k\}$. Note that the optimal solutions of both problems are at least as good as the optimal solution of $\max \{f(S) : |S| = k\}$. A careful correlation of the two instances preserves their approximation ratios, and allows us to combine their outputs into a solution for $\max \{f(S) : |S| = k\}$ achieving the same approximation ratio.
\inArXiv{\end{itemize}}\inJournal{\end{longitem}}

Our result for the problem $\max \{f(S) : S \subseteq \cN \}$ is based on a linear-time deterministic algorithm suggested by~\cite{BFNS15} for this problem. \citeWithName{Buchbinder et al.}{BFNS15} showed that this algorithm has an approximation ratio of $\nicefrac[]{1}{3}$ for general non-negative submodular functions. The algorithm maintains two solutions $X$ and $Y$ that become identical when the algorithm terminates. The analysis of the algorithm is based on a set $OPT(X, Y)$ that starts as $OPT$ and converts gradually to the final value of $X$ (and $Y$). The key observation of the analysis is showing that in each iteration (of the algorithm) the value of $OPT(X, Y)$ deteriorates by at most the increase in $f(X) + f(Y)$. In this work we show that the exact same algorithm provides $\nicefrac[]{1}{2}$-approximation for non-negative symmetric submodular functions. To that aim, we consider two sets $OPT(X, Y)$ and $\cOPT(X, Y)$. These sets start as $OPT$ and $\cOPT = \cN \setminus OPT$ respectively, and convert gradually into the final value of $X$ (and $Y$). We prove that the deterioration of $f(OPT(X, Y)) + f(\cOPT(X, Y))$ lower bounds the increase in $f(X) + f(Y)$.

\subsection{Related Work}
The literature on submodular maximization problems is very large, and therefore, we mention below only a few of the most relevant works.
\citeWithName{Feige et al.}{FMV11} provided the first constant factor approximation algorithms for $\max \{f(S) : S \subseteq \cN\}$. Their best approximation algorithm achieved an approximation ratio of $\nicefrac[]{2}{5} - o(1)$. \citeWithName{Oveis Gharan and Vondr\'{a}k}{GV11} used simulated annealing techniques to provide an improved approximation of roughly $0.41$. \citeWithName{Feldman et al.}{FNS11b} combined the algorithm of~\cite{GV11} with a new algorithm, yielding an approximation ratio of roughly $0.42$. Finally, \citeWithName{Buchbinder et al.}{BFNS15} gave a $\nicefrac[]{1}{2}$-approximation for this problem, matching a lower bound proved by~\cite{FMV11}.

The problem of maximizing a (not necessary monotone) submodular function subject to a general matroid constraint was given a $0.309$-approximation by~\cite{V13}.
Using simulated annealing techniques this was improved to $0.325$~\cite{GV11}, and shortly later was further pushed to $\nicefrac[]{1}{e}-o(1)$ by~\cite{FNS11} via the measured continuous greedy algorithm. Recently, \citeWithName{Buchbinder et al.}{BFNS14} showed that for the problem $\max \{f(S) : |S| \leq k\}$ (which is a special case of a matroid constraint) it is possible to get an approximation ratio in the range $[\nicefrac[]{1}{e} + \eps, \nicefrac[]{1}{2} - o(1)]$ for some small constant $\eps > 0$ (the exact approximation ratio in this range depends on the ratio $k / n$). A hardness result of $0.491$ was given by~\cite{GV11} for the case $k \ll n$.

The problem of maximizing a (not necessary monotone) submodular function subject to a matroid base constraint was shown to have no constant approximation ratio by~\cite{V13}. \citeWithName{Buchbinder et al.}{BFNS14} showed that the special case of $\max \{f(S) : |S| = k\}$ admits an approximation ratio in the range $[0.356, \nicefrac{1}{2} - o(1)]$ (again, the exact approximation ratio within this range depends on the ratio $k / n$). On the other hand, the hardness of $0.491$ by~\cite{GV11} applies also to this problem when $k \ll n$.

The {\SWfull} problem was studied in the case of monotone utility functions. The greedy algorithm achieves $\nicefrac[]{1}{2}$-approximation for this problem~\cite{FNW78}. This was improved to $1/(2 - \nicefrac[]{1}{k})$ by~\cite{DS06}, and then to $1 - \nicefrac[]{1}{e}$ by~\cite{CCPV11} using the celebrated continuous greedy algorithm. Finally, \citeWithName{Feldman et al.}{FNS11} gave a $\left(1 - (1 - \nicefrac[]{1}{k})^k\right)$-approximation algorithm, matching the hardness result of~\cite{MSV08}.
\section{Preliminaries} \label{sec:preliminaries}

For every set $S \subseteq \cN$ and an element $u \in \cN$, we denote the union $S \cup \{u\}$ by $S + u$, the expression $S \setminus \{u\}$ by $S - u$ and the set $\cN \setminus S$ by $\bar{S}$. Additionally, we use $\characteristic_S$ and $\characteristic_u$ to denote the characteristic vectors of $S$ and $\{u\}$, respectively. Given a submodular function $f\colon 2^\cN \to \bR$ and its corresponding multilinear extension $F\colon [0, 1]^\cN \to \bR$, we denote the partial derivative of $F$ at a point $x\in [0,1]^\cN$ with respect to an element $u$ by $\partial_u F(x)$. Since $F$ is multilinear, $\partial_u F(x) = F(x \vee \characteristic_u) - F(x \wedge \characteristic_{\cN - u})$. Additionally, we use $\bar{f}$ and $\bar{F}$ to denote the functions $\bar{f}(S) = f(\cN \setminus S)$ and $\bar{F}(x) = F(\characteristic_\cN - x)$. Finally, given a vector $x \in [0, 1]^\cN$, we denote $|x| = \sum_{u \in \cN} x_u$. 

We look for algorithms of  polynomial in $n$ (the size of $\cN$) time complexity. However, an explicit representation of a submodular function might be exponential in the size of its ground set. The standard way to bypass this difficulty is to assume access to the function via a \emph{value oracle}. For a submodular function $f\colon 2^\cN \to \bR$, given a set $S \subseteq \cN$, the value oracle returns the value of $f(S)$. Some of our algorithms assume a more powerful oracle that given a vector $x \in [0, 1]^\cN$, returns the value of $F(x)$. If such an oracle is not available, one can approximate it arbitrarily well using a value oracle to $f$ by averaging enough samples, which results in an $o(1)$ loss in the approximation ratio of the relevant algorithms (which has already been taken into account in the results presented in Section~\ref{sec:introduction}). This is a standard practice (see, \eg, \cite{CCPV11}), and we omit details.

The following lemma gives a few useful properties of submodular functions used throughout the paper.

\begin{lemma} \label{lem:general_properties}
If $f\colon 2^\cN \to \bR$ is a submodular function and $F\colon [0, 1]^\cN \to \bR$ is its multilinear extension, then:
\begin{\itemizeEnv}
	\item For every vector $x \in [0, 1]^\cN$, $\bar{F}(x) = F(\characteristic_\cN - x)$ is the multilinear extension of $\bar{f}$.
	\item If $f$ is symmetric, then for every vector $x \in [0, 1]^\cN$, $F(x) = \bar{F}(x)$.
	\item For every three vectors $z \leq y \leq x \in [0, 1]^\cN$, $F(x) - F(y) \leq F(x - z) - F(y - z)$.
\end{\itemizeEnv}
\end{lemma}
\ProofGeneralProperties

We are now ready to give the promised proof of Lemma~\ref{lem:union_bound_symmetric}.
\begin{replemma}{lem:union_bound_symmetric}
Given a non-negative symmetric submodular function $f\colon 2^\cN \to \bR^+$, a set $S \subseteq \cN$ and a vector $x \in [0, 1]$ obeying $F(y) \leq F(x)$ for every $\{y \in [0, 1]^\cN : y \leq x\}$, then $F(\characteristic_S \vee x) \geq f(S) - F(x)$.
\end{replemma}
\begin{proof}
Since $f$ is symmetric,
\[
	f(S) - F(x \vee \characteristic_S)
	=
	f(\bar{S}) - F((\characteristic_{\cN} - x) \wedge \characteristic_{\bar{S}})
	\leq
	F(x \wedge \characteristic_{\bar{S}}) - f(\varnothing)
	\leq
	F(x \wedge \characteristic_{\bar{S}})
	\leq
	F(x)
	\enspace,
\]
where the equality and first inequality hold by Lemma~\ref{lem:general_properties}, the second inequality holds by the non-negativity of $f$ and the last inequality holds since $x \wedge \characteristic_{\bar{S}} \leq x$.
\end{proof}

The following lemma shows that the multilinear extension behaves like a linear function within small neighborhoods. Similar lemmata appear in many works. A proof of this specific lemma can be found in~\cite{F13} (as Lemma~2.3.7).

\begin{lemma} \label{lem:greedy_step_increase_bound_general}
Consider two vectors $x, x' \in [0, 1]^\cN$ such that for every $u \in \cN$, $|x_u - x'_u| \leq \delta$. Then, $F(x') - F(x) \geq \sum_{u \in \cN} (x'_u - x_u) \cdot \partial_u F(x) - O(n^3\delta^2) \cdot \max_{u \in \cN} f(\{u\})$.
\end{lemma}

We also use the following lemma, which comes handy in proving the feasibility of the solutions produced by some of our algorithms. This lemma is implicitly proved by~\cite{FNS11} (some parts of the proof, which are omitted in~\cite{FNS11}, can be found in~\cite{F13}).

\begin{lemma} \label{lem:feasibility}
Fix some $\delta \leq n^{-5}$, and let $\{I(i)\}_{i = 1}^\ell$ be a set of $\ell$ points in a down-monotone polytope $\cP \subseteq [0, 1]^\cN$. Let $\{y(i)\}_{i = 0}^\ell$ be a a set of $\ell + 1$ vectors in $[0, 1]^\cN$ obeying the following constraints. For every element $u \in \cN$,
\[
	y_u(i) \leq
	\begin{cases}
		0 & \text{if $i = 0$} \enspace, \\
		y_u(i - 1) + \delta I_u(i) \cdot (1 - y_u(i - 1)) & \text{otherwise} \enspace.
	\end{cases}
\]
Then,
\begin{\itemizeEnv}
	\item $y_u(i) / (\delta i) \in \cP$.
	\item Let $T_{\cP} = -\ln(1 - d(\cP) + n^{-4})/d(\cP)$. Then, $\delta i \leq T_{\cP}$ implies $y(i) \in \cP$.
\end{\itemizeEnv}
\end{lemma}

\section{Measured Continuous Greedy for Symmetric Functions} \label{sec:symmetric_continuous_greedy}

In this section we prove Theorem~\ref{thm:polytope_symmetric}.

\begin{reptheorem}{thm:polytope_symmetric}
Given a non-negative symmetric submodular function $f\colon 2^\cN \to \bR^+$, a down-monotone solvable polytope $\cP \subseteq [0,1]^\cN$ and a constant $T \geq 0$, there exists an efficient algorithm that finds a point $x \in [0,1]^\cN$ such that $F(x) \geq \nicefrac[]{1}{2} \cdot [1 - e^{-2T} - o(1)] \cdot \max \{F(x) : x \in \cP \cap \{0,1\}^\cN\}$.
Additionally,
\begin{\enumerateEnv}[(a)]
    \item $x / T \in \cP$.
    \item Let $T_{\cP} = -\ln(1 - d(\cP) + n^{-4})/d(\cP)$. Then, $T \leq T_{\cP}$ implies $x \in \cP$.
\end{\enumerateEnv}
\end{reptheorem}

To simplify the proof of the theorem, we assume the following reduction was applied.

\begin{reduction} \label{re:singleton_ok}
We may assume in the proof of Theorem~\ref{thm:polytope_symmetric} that $\characteristic_u \in \cP$ for every $u \in \cN$.
\end{reduction}
\ProofSingeltonOK

The algorithm we use to prove Theorem~\ref{thm:polytope_symmetric} is Algorithm~\ref{alg:MeasuredContinuousGreedyKeepDown}, which is a variant of the Measured Continuous Greedy algorithm presented by~\cite{FNS11}. Notice that the definition of $\delta$ in the algorithm guarantees two properties: $\delta \leq n^{-5}$ and $t = T$ after $\lceil n^5 T \rceil$ iterations. These properties imply, by Lemma~\ref{lem:feasibility}, that the output of Algorithm~\ref{alg:MeasuredContinuousGreedyKeepDown} obeys properties~\eqref{cond:feasibility_1} and~\eqref{cond:feasibility_2} guaranteed by Theorem~\ref{thm:polytope_symmetric}. Thus, to complete the proof of Theorem~\ref{thm:polytope_symmetric}, it is only necessary to show that the approximation ratio of Algorithm~\ref{alg:MeasuredContinuousGreedyKeepDown} matches the approximation ratio guaranteed by the theorem.

\begin{algorithm}[!ht]
\caption{\textsf{Measured Continuous Greedy for Symmetric Functions}$(f, \cP, T)$} \label{alg:MeasuredContinuousGreedyKeepDown}
\DontPrintSemicolon
\tcp{Initialization}
Set: $\delta \gets T (\lceil n^5 T \rceil)^{-1}$.\\
Initialize: $t \gets 0$, $y(0) \gets \characteristic_\varnothing$.\\

\BlankLine

\tcp{Main loop}
\While{$t < T$}
{
    \lForEach{$u \in \cN$}
    {
        Let $w_u(t) \gets F(y(t) \vee \characteristic_u) - F(y(t))$.
    }
    Let $I(t)$ be a vector in $\cP$ maximizing $I(t) \cdot w(t)$.\\
    \lForEach{$u \in \cN$}
    {
        Let $y_u(t + \delta) \gets y_u(t) + \delta I_u(t) \cdot (1 - y_u(t))$.
    }
		\ForEach{$u \in \cN$}
		{ \label{line:second_scan}
			\lIf{$\partial_u F(y(t + \delta)) < 0$}{$y_u(t + \delta) \gets 0$.} \label{line:check_negative}
		}
}

\BlankLine

Return $y(T)$.
\end{algorithm}


First, we need a lower bound on the improvement achieved in each iteration of the algorithm. The following lemma is a counterpart of Lemma~III.2 of~\cite{FNS11}.

\LemmaStepApproximation

\begin{corollary} \label{cor:step_improvment}
For every time $0 \leq t < T$, $F(y(t + \delta)) - F(y(t)) \geq \delta \cdot [F(y(t) \vee \characteristic_{OPT}) - F(y(t))] - O(n^3\delta^2) \cdot f(OPT)$.
\end{corollary}
\ProofStepImprovement

The last corollary gives a lower bound on the improvement achieved in every step of the algorithm in terms of $F(y(t) \vee \characteristic_{OPT})$. To make this lower bound useful, we need to lower bound the term $F(y(t) \vee \characteristic_{OPT})$ using Lemma~\ref{lem:union_bound_symmetric}. The following lemma shows that the conditions of Lemma~\ref{lem:union_bound_symmetric} hold. 

\LemmaRemovingOnlyDecrease

\begin{corollary} \label{cor:good_solution}
For every time $0 \leq t < T$, $F(y(T + \delta)) - F(y(T)) \geq \delta \cdot [f(OPT) - 2 \cdot F(y(t))] - O(n^3\delta^2) \cdot f(OPT)$.
\end{corollary}
\ProofGoodSolution

At this point we have a lower bound on the improvement achieved in each iteration in terms of $f(OPT)$ and $F(y(t))$. In order to complete the analysis of the algorithm, we need to derive from it a bound on the value of $F(y(t))$ for every time $t$. \LemmataGH

We are now ready to prove the approximation ratio of Theorem~\ref{thm:polytope_symmetric} using the last lemmata.
\begin{proof}[\inArXiv{Proof }of the Approximation Ratio of Theorem~\ref{thm:polytope_symmetric}]
By Lemmata~\ref{lem:g_bound2} and \ref{lem:h_bound2},
\begin{align*}
    F(y(T))
    \geq{} &
    g(T) - O(n^3 \delta) \cdot T \cdot f(OPT)\\
    \geq{} &
    h(T) - O(n^3 \delta T) \cdot f(OPT)
    =
    (1/2) \cdot [1 - 2e^{-T} - O(n^3 \delta T)] \cdot f(OPT)
		\enspace.
\end{align*}
The proof is now complete since $T$ is a constant and $\delta \leq n^{-5}$.
\end{proof}
\section{Equality Cardinality Constraints} \label{sec:cardinality}

In this section we prove Theorem~\ref{thm:uniform_base_symmetric_approximation}.

\begin{reptheorem}{thm:uniform_base_symmetric_approximation}
There exists an efficient algorithm that given a non-negative symmetric submodular function $f\colon 2^\cN \to \bR^+$ and an integer cardinality parameter $1 \leq k \leq n/2$, achieves an approximation of $\nicefrac[]{1}{2}[1 - (1 - k/n)^{2n/k}] - o(1)$ for the problem: $\max \{f(S) : |S| = k\}$. If $k > n / 2$, then the same result holds with the cardinality parameter replaced by $n - k$.
\end{reptheorem}

The proof of Theorem~\ref{thm:uniform_base_general_approximation} is based on similar ideas, and is deferred to Appendix~\ref{sec:uniform_base_approximation_general}. To simplify the proof of Theorem~\ref{thm:uniform_base_symmetric_approximation}, we assume the following reduction was applied.

\begin{reduction} \label{re:base}
We may assume in the proof of Theorem~\ref{thm:uniform_base_symmetric_approximation} that $2k \leq n$.
\end{reduction}
\ProofReductionBase

The algorithm we use to prove Theorem~\ref{thm:uniform_base_symmetric_approximation} is Algorithm~\ref{alg:DoubleContinuousGreedy}. One can think of this algorithm as two synchronized instances of Algorithm~\ref{alg:MeasuredContinuousGreedyKeepDown}. One instance starts with the solution $\characteristic_\varnothing$ and looks for a solution obeying the constraint $\sum_{u \in \cN} x_u \leq k$. The other instance starts with the solution $\characteristic_\cN$ and looks for a solution obeying the constraint $\sum_{u \in \cN} x_u \geq k$ (alternatively, we can think of the second instance as having the objective $\bar{f}$ and the constraint $\sum_{u \in \cN} x_u \leq n - k$). The two instances are synchronized in two senses:
\begin{\itemizeEnv}
	\item In each iteration, the two instances choose direction vectors $I^1$ and $I^2$ obeying $I^1 + I^2 = \characteristic_\cN$ (\ie, the direction vector of one instance implies the direction vector of the other instance).
	\item The direction vectors are selected in a way that improves the solutions of both instances.
\end{\itemizeEnv}

The output of Algorithm~\ref{alg:DoubleContinuousGreedy} is a fractional solution. This solution can be rounded into an integral solution using a standard rounding procedure such as pipage rounding~\cite{CCPV11}.

\begin{algorithm}[!ht]
\caption{\textsf{Double Measured Continuous Greedy}$(f, \cN, k)$} \label{alg:DoubleContinuousGreedy}
\DontPrintSemicolon
\tcp{Initialization}
Set: $T \gets -\nicefrac[]{n}{k} \cdot \ln(1 - k/n + n^{-4})$ and $\delta \gets T (\lceil n^5 T \rceil)^{-1}$.\\
Initialize: $t \gets 0$, $y^1(0) \gets \characteristic_\varnothing$ and $y^2(0) \gets \characteristic_\cN$.\\

\BlankLine

\tcp{Main loop}
\While{$t < T$}
{
    \ForEach{$u \in \cN$}
    {
        Let $w^1_u(t) \gets F(y^1(t) \vee \characteristic_u) - F(y^1(t))$ and $w^2_u(t) \gets F(y^2(t) \wedge \characteristic_{\cN - u}) - F(y^2(t))$.
    }
    Let $I^1(t) \in [0, 1]^\cN$ and $I^2(t) \in [0, 1]^\cN$ be two vectors maximizing
		\[
			\min\{I^1(t) \cdot w^1(t) + 2 \cdot F(y^1(t)), I^2(t) \cdot w^2(t) + 2 \cdot F(y^2(t)) \}
		\]
		among the vectors obeying $|I^1(t)| = k$, $|I^2(t)| = n - k$ and $I^1(t) + I^2(t) = \characteristic_\cN$. \label{line:calc_I}\\
    \ForEach{$u \in \cN$}
    {
        Let $y^1_u(t + \delta) \gets y^1_u(t) + \delta I^1_u(t) \cdot (1 - y^1_u(t))$ and $y^2_u(t + \delta) \gets y^2_u(t) - \delta I^2_u(t) \cdot y^2_u(t)$.
    }
		\ForEach{$u \in \cN$}
		{ \label{line:second_scan_two_sided}
			\lIf{$\partial_u F(y^1(t + \delta)) < 0$}{$y^1_u(t + \delta) \gets 0$.}
			\lIf{$\partial_u F(y^2(t + \delta)) > 0$}{$y^2_u(t + \delta) \gets 1$.}
		}
    $t \leftarrow t + \delta$.
}

\BlankLine

\lIf{$|y^1(T)| = |y^2(T)|$}
{
	\Return{$y^1(T)$}.
}
\lElse
{
	\Return{$y^1(T) \cdot \frac{|y^2(T)| - k}{|y^2(T)| - |y^1(T)|} + y^2(T) \cdot \frac{k - |y^1(T)|}{|y^2(T)| - |y^1(T)|}$}.
}
\end{algorithm}

We begin the analysis of Algorithm~\ref{alg:DoubleContinuousGreedy} by showing it can be implemented efficiently using an LP solver.

\begin{observation}
There exists an efficient algorithm for calculating the vectors $I^1(t)$ and $I^2(t)$ defined on Line~\ref{line:calc_I} of Algorithm~\ref{alg:DoubleContinuousGreedy}.
\end{observation}
\begin{proof} \belowdisplayskip=-12pt
The calculation of $I^1(t)$ and $I^2(t)$ can be done by solving the following linear program.
\[
	\begin{array}{llll}
		\max & m \\
		\mbox{s.t.} & \sum_{u \in \cN} w^i_u(t) \cdot I^i_u(t) + 2 \cdot F(y^i(t)) & \geq m & \forall i \in \{1, 2\} \\
		& \sum_{u \in \cN} I^1_u(t) & = k &\\
		& \sum_{u \in \cN} I^2_u(t) & = n - k &\\
		& I^1_u(t) + I^2_u(t) & = 1 & \forall \, u \in \cN\\
		& I^i_u(t) & \geq 0 & \forall \, u \in \cN, i \in \{1, 2\}
	\end{array}
\]
\end{proof}

The following lemma follows from Lemma~\ref{lem:feasibility}.

\begin{lemma} \label{lem:y_properties}
For every time $0 \leq t \leq T$, the vectors $y^1(t)$ and $y^2(t)$ obey:
\begin{\itemizeEnv}
	\item $y^1(t),y^2(t) \in [0, 1]^\cN$.
	\item $y^1(t) \leq y^2(t)$ (element-wise).
	\item $|y^1(t)| \leq k \leq |y^2(t)|$.
\end{\itemizeEnv}
\end{lemma}
\ProofYProperties

As a corollary of Lemma~\ref{lem:y_properties}, we can guarantee feasibility. Let $y$ be the vector produced by Algorithm~\ref{alg:DoubleContinuousGreedy}.

\begin{corollary} \label{cor:feasibility}
$y$ is a feasible solution.
\end{corollary}
\ProofFeasiblity

Our next objective is lower bounding $F(y)$ in terms of $F(y^1(T))$ and $F(y^2(T))$. Let $r\colon [0, 1] \to \bR^+$ be the function:
\[
	r(x) = F(y^1(T) + x(y^2(T) - y^1(T)))
	\enspace.
\]
Intuitively, $r(x)$ evaluates $F$ on a vector that changes from $y^1(T)$ to $y^2(T)$ as $x$ increases.

\begin{observation} \label{obs:concave}
$r$ is a non-negative concave function.
\end{observation}
\ProofConcave

\begin{corollary} \label{cor:value_as_worse_option}
$F(y) \geq \min\{F(y^1(T)), F(y^2(T))\}$.
\end{corollary}
\begin{proof}
If $|y^1(T)| = |y^2(T)|$ then $y = y^1(T)$, which makes the corollary trivial. Thus, we may assume from now on: $|y^1(T)| \neq |y^2(T)|$. Observe that in this case:
\begin{align*}
	F(y)
	={} &
	F\left(y^1(T) \cdot \frac{|y^2(T)| - k}{|y^2(T)| - |y^1(T)|} + y^2(T) \cdot \frac{k - |y^1(T)|}{|y^2(T)| - |y^1(T)|}\right)\\
	={} &
	F\left(y^1(T) + (y^2(T) - y^1(T)) \cdot \frac{k - |y^1(T)|}{|y^2(T)| - |y^1(T)|}\right)
	=
	r\left(\frac{k - |y^1(T)|}{|y^2(T)| - |y^1(T)|}\right)
	\enspace.
\end{align*}
Notice that $(k - |y^1(T)|) / (|y^2(T)| - |y^1(T)|) \in [0, 1]$. Thus, the concavity of $r$ implies:
\[
	F(y)
	=
	r\left(\frac{k - |y^1(T)|}{|y^2(T)| - |y^1(T)|}\right)
	\geq
	\min\{r(0), r(1)\}
	=
	\min\{F(y^1(T)), F(y^2(T))\}
	\enspace.
	\qedhere
\]
\end{proof}

The proof of Theorem~\ref{thm:uniform_base_symmetric_approximation} now boils down to lower bounding $\min\{F(y^1(T)), F(y^2(T))\}$. The following lemma is a counter-part of Lemma~\ref{lem:step_approximation}. Let $\Delta(t) = \min\{F(y^1(t) \vee \characteristic_{OPT}) + F(y^1(t)), F(y^2(t) \wedge \characteristic_{OPT}) + F(y^2(t))\}$.

\begin{lemma} \label{lem:step_approximation_two_sided}
For every time $0 \leq t < T$:
\[
	\sum_{u \in \cN} (1 - y^1_u(t)) \cdot I^1_u(t) \cdot \partial_u F(y^1(t)) + 2 \cdot F(y^1(t)) \geq \Delta(t)
	\enspace,
\]
\[
	-\sum_{u \in \cN} y^2_u(t) \cdot I^2_u(t) \cdot \partial_u F(y^2(t)) + 2 \cdot F(y^2(t)) \geq \Delta(t)
	\enspace.
\]
\end{lemma}
\ProofStepApproximationTwoSided

\begin{lemma} \label{le:OPT_not_too_bad}
$f(OPT) \geq \max_{u \in \cN} f(\{u\}) / 2$.
\end{lemma}
\ProofOptNotBad

\begin{corollary} \label{cor:step_improvment_two_sided}
For every time $0 \leq t < T$,
\[
	F(y^1(t + \delta)) - F(y^1(t)) \geq \delta \cdot [\Delta(t) - 2 \cdot F(y^1(t))] - O(n^3\delta^2) \cdot f(OPT)
	\enspace,
\]
and
\[
	F(y^2(t + \delta)) - F(y^2(t)) \geq \delta \cdot [\Delta(t) - 2 \cdot F(y^2(t))] - O(n^3\delta^2) \cdot f(OPT)
	\enspace.
\]
\end{corollary}
\begin{proof}
Lemmata~\ref{lem:greedy_step_increase_bound_general}, \ref{lem:step_approximation_two_sided} and~\ref{le:OPT_not_too_bad} imply:
\[
	F(y^1(t + \delta)) - F(y^1(t))
	\geq
	\delta \cdot [\Delta(t) - 2 \cdot F(y^1(t))] - O(n^3\delta^2) \cdot f(OPT)
	\enspace,
\]
where $y^1(t + \delta)$ represents its value at the beginning of the loop starting on Line~\ref{line:second_scan_two_sided}. The first part of the corollary now follows by noticing that the last loop can only increase the value of $F(y^1(t + \delta))$. The second part of the corollary is analogous.
\end{proof}

To make the lower bounds given by the above lemma useful, we need a lower bound on $\Delta(t)$. This lower bound is obtained using Lemma~\ref{lem:union_bound_symmetric}. Proving that the conditions of Lemma~\ref{lem:union_bound_symmetric} hold can be done using ideas from the proof of Lemma~\ref{lem:removing_only_decrease}.

\begin{lemma} \label{lem:delta_bound}
For every time $0 \leq t < T$, $\Delta(t) \geq f(OPT)$.
\end{lemma}
\ProofDeltaBound

Corollary~\ref{cor:step_improvment_two_sided} and Lemma~\ref{lem:delta_bound} imply together the following counterpart of Corollary~\ref{cor:good_solution}.
\begin{corollary}
For every time $0 \leq t < T$,
\[
	F(y^1(T + \delta)) - F(y^1(T)) \geq \delta \cdot [f(OPT) - 2 \cdot F(y^1(t))] - O(n^3\delta^2) \cdot f(OPT)
	\enspace,
\]
and
\[
	F(y^2(T + \delta)) - F(y^2(T)) \geq \delta \cdot [f(OPT) - 2 \cdot F(y^2(t))] - O(n^3\delta^2) \cdot f(OPT)
	\enspace.
\]
\end{corollary}

Repeating the same line of arguments used in Section~\ref{sec:symmetric_continuous_greedy}, the previous corollary implies:
\begin{lemma} \label{lem:values_symmetric}
$F(y^1(T)) \geq \nicefrac[]{1}{2} \cdot [1 - e^{-2T} - o(1)] \cdot f(OPT)$ and $F(y^2(T)) \geq \nicefrac[]{1}{2} \cdot [1 - e^{-2T} - o(1)] \cdot f(OPT)$.
\end{lemma}

We are now ready to prove the approximation ratio guaranteed by Theorem~\ref{thm:uniform_base_symmetric_approximation}.

\begin{proof}[\inArXiv{Proof }of the Approximation Ratio of Theorem~\ref{thm:uniform_base_symmetric_approximation}]
By Corollary~\ref{cor:value_as_worse_option} and Lemma~\ref{lem:values_symmetric}, the approximation ratio of Algorithm~\ref{alg:DoubleContinuousGreedy}, up to an error term of $o(1)$, is at least:
\begin{align*}
	\frac{1 - e^{-2[-(n/k) \cdot \ln(1 - k/n + n^{-4})]}}{2}
	={} &
	\frac{1 - (1 - k/n)^{2n/k} \cdot \left(1 + \frac{n^{-4}}{1 - k/n}\right)^{2n/k}}{2}\inJournal{\\}
	\geq\inJournal{{} &}
	\frac{1 - (1 - k/n)^{2n/k} \cdot e^{4n^{-3}k^{-1}}}{2}\\
	\geq{} &
	\frac{1 - (1 - k/n)^{2n/k} - [e^{4n^{-3}k^{-1}} - 1]}{2}\inJournal{\\}
	=\inJournal{{} &}
	\frac{1 - (1 - k/n)^{2n/k}}{2} - o(1)
	\enspace.
	\qedhere
\end{align*}
\end{proof}
\section{Deterministic $1/2$-Approximation for Unconstrained Symmetric Submodular Maximization} 

In this section we prove Theorem~\ref{thm:deterministic_symmetric_unconstrained}.

\begin{reptheorem}{thm:deterministic_symmetric_unconstrained}
There exists a deterministic linear-time $\nicefrac[]{1}{2}$-approximation algorithm for the problem $\max \{f(S) : S \subseteq \cN\}$, where $f\colon 2^\cN \to \bR^+$ is a non-negative symmetric submodular function.
\end{reptheorem}

Algorithm~\ref{alg:UnconstrainedDeterminisitic} is a restatement of Algorithm~1 of~\cite{BFNS15}. \citeWithName{Buchbinder et al.}{BFNS15} proved that this algorithm provides a $1/3$-approximation for the problem $\max \{f(S) : S \subseteq \cN\}$ when $f$ is a non-negative submodular function. Moreover, they showed a tight example for which the algorithm achieves only $1/3 + \eps$ approximation. We prove that when $f$ is also symmetric, the approximation ratio of this algorithm improves to $1/2$, and thus, prove Theorem~\ref{thm:deterministic_symmetric_unconstrained}.


\begin{algorithm}[!ht]
\caption{\textsf{Two-Sided Greedy}$(f,\cN)$} \label{alg:UnconstrainedDeterminisitic}
\DontPrintSemicolon
$X_0 \leftarrow \emptyset$, $Y_0 \leftarrow \cN$.\\
\For{$i$ = $1$ \KwTo $n$}{
  $a_i \gets f(X_{i - 1} + u_i) - f(X_{i - 1})$.\\
  $b_i \gets f(Y_{i - 1} - u_i) - f(Y_{i - 1})$.\\
  \lIf{$a_i \geq b_i $}{$X_i \gets X_{i - 1} + u_i$, $Y_i \gets Y_{i - 1}$.}
  \lElse{$X_i \gets X_{i - 1}$, $Y_i \gets Y_{i - 1} - u_i$.}
}
\Return{$X_n$} (or equivalently $Y_n$).
\end{algorithm}

Following~\cite{BFNS15}, we define $OPT_i \triangleq (OPT \cup X_i) \cap Y_i$, \ie, $OPT_i$ agrees with $X_i$ (and $Y_i$) on the first $i$ elements and with $OPT$ on the last $n-i$ elements. Similarly, we also define $\cOPT_i \triangleq (\cOPT \cup X_i) \cap Y_i$.

\begin{observation} \label{obs:basic_equalities}
$OPT_0 = OPT$, $\cOPT_0 = \cOPT$ and $OPT_n = \cOPT_n = X_n = Y_n$.
\end{observation}

Consider the change in the value of $f(OPT_i) + f(\cOPT_i)$ as $i$ increases. The value of this expression starts as $2f(OPT)$ (for $i = 0$) and deteriorates to $2f(X_n)$ (for $i = n$). The main idea of the proof is to bound the total loss of value. This goal is achieved by the following lemma which upper bounds the loss in value occurring whenever $i$ increases by $1$. More formally, the lemma shows that the decrease in $f(OPT_i) + f(\cOPT_i)$ is no more than the \emph{total} increase in value of both solutions maintained by the algorithm, \ie, $f(X_i) + f(Y_i)$.

\begin{lemma} \label{lem:loss_gain_bound}
For every $1 \leq i \leq n$,
\[
	[f(OPT_{i - 1}) - f(OPT_i)] + [f(\cOPT_{i - 1}) - f(\cOPT_i)]
	\leq
	[f(X_i) - f(X_{i - 1})] + [f(Y_i) - f(Y_{i - 1})]
	\enspace.
\]
\end{lemma}

Before proving Lemma~\ref{lem:loss_gain_bound}, let us show that Theorem~\ref{thm:deterministic_symmetric_unconstrained} follows from it.

\begin{proof}[\inArXiv{Proof }of Theorem~\ref{thm:deterministic_symmetric_unconstrained}]
Adding up Lemma~\ref{lem:loss_gain_bound} for every $1 \leq i \leq n$ gives:
\begin{align*}
  \sum_{i = 1}^n [f(OPT_{i - 1}) - f(OPT_i)] &+ \sum_{i = 1}^n [f(\cOPT_{i - 1}) - f(\cOPT_i)]\\
  \leq{} &
  \sum_{i = 1}^n [f(X_i) - f(X_{i - 1})] + \sum_{i = 1}^n [f(Y_i) - f(Y_{i - 1})]
	\enspace.
\end{align*}
The above sums are telescopic. Collapsing them and using the non-negativity of $f$ results in:
\begin{align*}
  [f(OPT_0) - f(OPT_n)] + [f(\cOPT_0) - f(\cOPT_n)]
  \leq{} &
  [f(X_n) - f(X_0)] + [f(Y_n) - f(Y_0)]\\
  \leq{} &
  f(X_n) + f(Y_n) \enspace.
\end{align*}
Using the equalities of Observation~\ref{obs:basic_equalities}, we obtain:
\[
	[f(OPT) - f(X_n)] + [f(\cOPT) - f(X_n)] \leq f(X_n) + f(X_n)
	\Rightarrow
	f(X_n) = \frac{f(OPT) + f(\cOPT)}{4}
	\enspace.
\]
The theorem now follows from the symmetry of $f$.
\end{proof}

It all boils down now to proving Lemma~\ref{lem:loss_gain_bound}.

\begin{proof}[\inArXiv{Proof }of Lemma~\ref{lem:loss_gain_bound}]
Assume $a_i \geq b_i$ (the other case is analogous). This assumption implies $X_i = X_{i-1} + u_i$ and $Y_i = Y_{i-1}$, which induce:
\begin{itemize}
 \item $OPT_i = (OPT \cup X_i) \cap Y_i = OPT_{i-1} + u_i$.
 \item $\cOPT_i = (\cOPT \cup X_i) \cap Y_i = \cOPT_{i-1} + u_i$.
\end{itemize}
Thus, the lemma we want to prove can be rewritten as:
\begin{equation} \label{eq:lemma_boils_to}
	[f(OPT_{i - 1}) - f(OPT_{i-1} + u_i)] + [f(\cOPT_{i - 1}) - f(\cOPT_{i - 1} + u_i)]
	\leq
	f(X_i) - f(X_{i - 1})
	=
	a_i
	\enspace.
\end{equation}

We have to consider two cases. If $u_i \in OPT$, then the left side of Equation~\eqref{eq:lemma_boils_to} is equal to:
\[
	f(\cOPT_{i - 1}) - f(\cOPT_{i - 1} + u_i)
	\leq
	f(Y_{i-1} - u_i) - f(Y_{i - 1})
	=
	b_i
	\leq
  a_i
	\enspace,
\]
where the first inequality follows by submodularity: $\overline{OPT}_{i-1} = (\cOPT \cup X_{i-1}) \cap Y_{i-1} \subseteq Y_{i-1} - u_i$ since $u_i \not \in \cOPT \cup X_{i-1}$. If $u_i \not \in OPT$, then the left side of Equation~\eqref{eq:lemma_boils_to} is equal to:
\[
	f(OPT_{i - 1}) - f(OPT_{i - 1} + u_i)
	\leq
	f(Y_{i-1} - u_i) - f(Y_{i - 1})
	=
	b_i
	\leq
  a_i
	\enspace,
\]
where the first inequality follows, again, by submodularity: $OPT_{i-1} = (OPT \cup X_{i-1}) \cap Y_{i-1} \subseteq Y_{i-1} - u_i$ since $u_i \not \in OPT \cup X_{i-1}$.
\end{proof}
\section{Submodular Welfare with Identical Utilities}

In this section we prove Theorem~\ref{thm:identical_submodular_welfare}. The positive and negative parts of the theorem are proved in Sections~\ref{ssc:positive_welfare} and~\ref{ssc:negative_welfare}, respectively.

\begin{reptheorem}{thm:identical_submodular_welfare}
There exists a linear-time $[1 - (1 - 1/k)^{k - 1}]$-approximation algorithm for {\SW} with $k$ players having identical non-negative submodular utility functions. Moreover, any algorithm for this problem whose approximation ratio is $[1 - (1 - 1/k)^{k - 1}] + \eps$ (for some constant $\eps > 0$) must use an exponential number of value oracle queries.
\end{reptheorem}

\subsection{Proof of the Positive Part of Theorem~\ref{thm:identical_submodular_welfare}} \label{ssc:positive_welfare}

Consider the algorithm assigning every element $u \in \cN$ to a random one out of the $k$ players. A formal description of this algorithm is given as Algorithm~\ref{alg:IdenticalSW} (the notation $[k]$ used by Algorithm~\ref{alg:IdenticalSW} denotes the set $\{1, 2, \dotsc, k\}$). We show that Algorithm~\ref{alg:IdenticalSW} has the approximation ratio guaranteed by Theorem~\ref{thm:identical_submodular_welfare}.

\begin{algorithm}[!ht]
\caption{\textsf{Random Assignment}$(f,k,\cN)$} \label{alg:IdenticalSW}
\DontPrintSemicolon
\For{$i$ = $1$ \KwTo $k$}
{
	Initialize $S_i \gets \varnothing$.\\
}
\For{each element $u \in \cN$}
{
  Choose a uniformly random $i \in [k]$.\\
  Update $S_i \gets S_i + u$.\\
}
\For{$i$ = $1$ \KwTo $k$}
{
	Assign the elements of $S_i$ to player $p_i$.
}
\end{algorithm}

Despite the simplicity of Algorithm~\ref{alg:IdenticalSW}, we do not have a simple analysis of its approximation ratio making it intuitively clear why the approximation ratio is what it is. Instead we give two analyses which prove this approximation ratio through, somewhat involved and unintuitive, mathematical manipulations. One analysis of the algorithm can be found in Appendix~\ref{app:SWOriginal}. Below we give a quite different simpler analysis suggested by Uri Feige. Both analyses use the two following known lemmata.

\begin{lemma}[{\rm Lemma~2.2 of \cite{FMV11}}] \label{lem:equal_prob_bound}
Let $f\colon 2^\cN \to \bR$ be submodular. Denote by $A(p)$ a random subset of $A$ where each element appears with probability $p$. Then, $\bE[f(A(p))] \geq (1 - p) f(\varnothing) + p \cdot f(A)$.
\end{lemma}

Following the notation of Lemma~\ref{lem:equal_prob_bound}, given a set $A$ and a probability $p$, we use $A(p)$ in the rest of the paper to denote a random set containing every element of $A$ with probability $p$, independently.


\begin{lemma}[{\rm Lemma~2.2 of~\cite{BFNS14} (rephrased)}] \label{lem:max_probability_max_damage}
Let $f\colon 2^\cN \to \bR^+$ be a non-negative submodular function, and let $R$ be a random set in which each element appears with probability at most $p$ (not necessarily independently). Then, $\bE[f(R)] \geq (1 - p) f(\varnothing)$.
\end{lemma}

For every $1 \leq i \leq k$, let $OPT_i$ denote the set of elements assigned by the optimal solution to player $p_i$. Additionally, let $\pi\colon [k] \to [k]$ be a uniformly random permutation of $[k]$; and for every $0 \leq i \leq k$, let $T_i$ be the set of the elements assigned by the optimal solution to the first $i$ players according to the order defined by the permutation $\pi$. More formally,
\[
	T_i
	=
	\bigcup_{j = 1}^i OPT_{\pi(j)}
	\qquad
	\forall\; 0 \leq i \leq k
	\enspace.
\]

The following lemma bounds the change in $\bE[f(T_i(k^{-1}))]$ as a function of $i$. Let $\opt$ be the value of the optimal solution (\ie, $\opt = \sum_{i = 1}^k f(OPT_i)$).

\begin{lemma} \label{lem:partial_union_recursive}
For every $1 \leq i \leq k$,
\[
	\bE[f(T_i(k^{-1}))]
	\geq
	\left(1 - \frac{1}{k}\right) \cdot \bE[f(T_{i-1}(k^{-1}))] + \left(1 - \frac{i - 1}{k(k + 1)}\right) \cdot \frac{\opt}{k^2}
	\enspace.
\]
\end{lemma}
\begin{proof}
Let us fix the permutation $\pi$ and the set $T_{i - 1}(k^{-1})$. Observe that after these fixes each element of $OPT_{\pi(i)}$ belongs to $T_{i - 1}(k^{-1}) \cup OPT_{\pi(i)}(k^{-1})$ with probability $k^{-1}$. Moreover, $f(T_{i - 1}(k^{-1}) \cup S)$ is a non-negative submodular function of $S$. Hence, by Lemma~\ref{lem:equal_prob_bound}:
\[
	\bE[f(T_{i - 1}(k^{-1}) \cup OPT_{\pi(i)}(k^{-1}))]
	\geq
	\left(1 - \frac{1}{k}\right) \cdot f(T_{i - 1}(k^{-1})) + \frac{1}{k} \cdot f(T_{i - 1}(k^{-1}) \cup OPT_{\pi(i)})
	\enspace.
\]

We now unfix the set $T_{i - 1}(k^{-1})$ and the permutation $\pi$, except for $\pi(i)$. By the law of total expectation, the previous inequality now becomes:
\begin{equation} \label{eq:mostly_unfixed}
	\bE[f(T_{i - 1}(k^{-1}) \cup OPT_{\pi(i)}(k^{-1}))]
	\geq
	\left(1 - \frac{1}{k}\right) \cdot \bE[f(T_{i - 1}(k^{-1}))] + \frac{1}{k} \cdot \bE[f(T_{i - 1}(k^{-1}) \cup OPT_{\pi(i)})]
	\enspace,
\end{equation}
where the expectations are over the random choice of the entries other than $\pi(i)$ in $\pi$, the subset of the elements of $T_{i - 1}$ that remain in $T_{i - 1}(k^{-1})$ and the subset of the elements of $OPT_{\pi(i)}$ that remain in $OPT_{\pi(i)}(k^{-1}))$.

Observe that an element $u \in \cN$ belongs to $T_{i - 1}$ if and only if it belongs to one of the sets $\{OPT_{\pi(j)}\}_{j = 1}^{i - 1}$, which happens with probability at most $\frac{i - 1}{k - 1}$ (we say ``at most'' since this probability is $0$ for elements of $OPT_{\pi(i)}$). Moreover, notice that $OPT_{\pi(i)}$ is deterministic (as we are still fixing $\pi(i)$), and $f(S \cup OPT_{\pi(i)})$ is a non-negative submodular function of $S$. Hence, by Lemma~\ref{lem:max_probability_max_damage},
\begin{align*}
	\bE[f(T_{i - 1}(k^{-1}) \cup OPT_{\pi(i)})]
	\geq{} &
	\left(1 - \frac{\max_{u \in \cN} \Pr[u \in T_{i - 1}]}{k}\right) \cdot f(OPT_{\pi(i)})\\
	\geq{} &
	\left(1 - \frac{i - 1}{k(k - 1)}\right) \cdot f(OPT_{\pi(i)})
	\enspace.
\end{align*}

Plugging the last inequality into Inequality~\eqref{eq:mostly_unfixed} and unfixing $\pi(i)$, we get:
\begin{align*}
	\bE[f(T_{i - 1}(k^{-1}) \cup \inJournal{{} &}OPT_{\pi(i)}(k^{-1}))]\inJournal{\\}
	\geq{} &
	\left(1 - \frac{1}{k}\right) \cdot \bE[f(T_{i - 1}(k^{-1}))] + \left(1 - \frac{i - 1}{k(k - 1)}\right) \cdot \frac{\bE[f(OPT_{\pi(i)})]}{k}\\
	={} &
	\left(1 - \frac{1}{k}\right) \cdot \bE[f(T_{i - 1}(k^{-1}))] + \left(1 - \frac{i - 1}{k(k - 1)}\right) \cdot \frac{\opt}{k^2}
	\enspace.
\end{align*}

The lemma now follows since $T_i(k^{-1})$ has the same distribution as $T_{i - 1}(k^{-1}) \cup OPT_{\pi(i)}(k^{-1})$.
\end{proof}

Lemma~\ref{lem:partial_union_recursive} gives a recursive formula for a lower bound on $\bE[f(T_i(k^{-1}))]$. The next lemma proves a closed form of this lower bound.

\begin{lemma} \label{lem:partial_union_closed}
For every $0 \leq i \leq k$,
\[
	\bE[f(T_i(k^{-1}))]
	\geq
	\left[\frac{k^2 - i}{k(k - 1)} - \left(1 - \frac{1}{k}\right)^{i - 1}\right] \cdot \frac{\opt}{k}
	\enspace.
\]
\end{lemma}
\begin{proof}
We prove the lemma by induction on $i$. First, let us prove the base case. Since $f$ is non-negative:
\[
	\bE[f(T_0(k^{-1}))]
	=
	f(\varnothing)
	\geq
	0
	=
	\left[\frac{k}{k - 1} - \frac{k}{k - 1}\right] \cdot \frac{\opt}{k}
	=
	\left[\frac{k^2 - 0}{k(k - 1)} - \left(1 - \frac{1}{k}\right)^{0 - 1}\right] \cdot \frac{\opt}{k}
	\enspace.
\]
Next, assume the lemma holds for $i - 1 \geq 0$, and let us prove it for $i$. By Lemma~\ref{lem:partial_union_recursive} and the induction hypothesis,
\begin{align*}
	\bE[f(T_i(k^{-1}))]
	\geq{} &
	\left(1 - \frac{1}{k}\right) \cdot \bE[f(T_{i - 1}(k^{-1}))] + \left(1 - \frac{i - 1}{k(k - 1)}\right) \cdot \frac{\opt}{k^2}\\
	\geq{} &
	\left(1 - \frac{1}{k}\right) \cdot \left[\frac{k^2 - i + 1}{k(k - 1)} - \left(1 - \frac{1}{k}\right)^{i - 2}\right] \cdot \frac{\opt}{k} + \left(1 - \frac{i - 1}{k(k - 1)}\right) \cdot \frac{\opt}{k^2}\\
	={} &
	\left[\frac{(1 - 1/k)(k^2 - i + 1) + (k - 1) - (i - 1)/k}{k(k - 1)} - \left(1 - \frac{1}{k}\right)^{i - 1}\right] \cdot \frac{\opt}{k}
	\enspace.
\end{align*}

The lemma now follows by plugging the next equality into the previous inequality.
\[
	\left(1 - \frac{1}{k}\right)(k^2 - i + 1) + (k - 1) - \frac{i - 1}{k}
	=
	(k^2 - i + 1) - k + \frac{i - 1}{k} + (k - 1) - \frac{i - 1}{k}
	=
	k^2 - i
	\enspace.
	\qedhere
\]
\end{proof}

We are now ready to prove the positive part of Theorem~\ref{thm:identical_submodular_welfare}.

\begin{proof}[\inArXiv{Proof }of the Positive Part of Theorem~\ref{thm:identical_submodular_welfare}]
Observe that $T_k(k^{-1})$ is a random set containing every element of $\cN$ with probability $k$, independently. Hence, $T_k(k^{-1})$ has the same distribution as every one of the sets $\{S_i\}_{i = 1}^k$. Thus, by Lemma~\ref{lem:partial_union_closed}:
\begin{align*}
	\bE\mathopen{}\left[\sum_{i = 1}^k f(S_i)\right]
	={} &
	k \cdot
	\bE[f(T_k(k^{-1}))]\\
	\geq{} &
	k \cdot \left[\frac{k^2 - k}{k(k - 1)} - \left(1 - \frac{1}{k}\right)^{k - 1}\right] \cdot \frac{\opt}{k}
	=
	\left[1 - \left(1 - \frac{1}{k}\right)^{k - 1}\right] \cdot \opt
	\enspace.
\end{align*}
The theorem now follows since $\sum_{i = 1}^k f(S_i)$ is the value of the solution produced by Algorithm~\ref{alg:IdenticalSW}.
\end{proof}

\subsection{Proof of the Negative Part of Theorem~\ref{thm:identical_submodular_welfare}} \label{ssc:negative_welfare}

Let us begin with a tight example showing that our analysis of Algorithm~\ref{alg:IdenticalSW} is tight. Consider an instance of {\SW} with $k \geq 2$ players and a set $\cN$ of $k$ items. The utility function of all the players is the non-negative submodular function $f(S)\colon \cN \to \bR^+$ defined as follows.
\[
	f(S)
	=
	\begin{cases}
		1 - \frac{|S| - 1}{k - 1} & \text{if $S \neq \varnothing$} \enspace,\\
		0 & \text{otherwise} \enspace.
	\end{cases}
\]

\begin{observation}
$\bE[f(\cN(1/k))] = 1 - (1 - 1/k)^{k - 1}$.
\end{observation}
\begin{proof}
Observe that:
\begin{align*}
	\bE[f(\cN(1/k))]
	={} &
	\bE\left[1 - \frac{|\cN(1/k)| - 1}{k - 1}\right] - \Pr[\cN(1/k) = \varnothing] \cdot \left[1 - \frac{0 - 1}{k - 1}\right]\\
	={} &
	1 - \frac{\bE[|\cN(1/k)|] - 1}{k - 1} - (1 - 1/k)^k \cdot \frac{k}{k - 1}
	=
	1 - (1 - 1/k)^{k - 1}
	\enspace.
	\qedhere
\end{align*}
\end{proof}

\begin{corollary}
There exists an instance of {\SW} with $k$ players having identical non-negative submodular utility functions for which the approximation ratio of Algorithm~\ref{alg:IdenticalSW} is $1 - (1 - 1/k)^{k - 1}$.
\end{corollary}
\begin{proof}
The above instance of {\SW} has an optimal solution $OPT$ assigning a single (arbitrary but unique) element to every player. The value of this solution is $k$. On the other hand, the solution produced by Algorithm~\ref{alg:IdenticalSW} has an expected value of:
\[
	\bE\left[\sum_{i = 1}^k f(S_i)\right]
	=
	k \cdot \bE[f(\cN(1/k))]
	=
	k[1 - (1 - 1/k)^{k - 1}]
	\enspace.
	\qedhere
\]
\end{proof}

In order to convert the above tight example into an hardness result, we need the following lemma from~\cite{V13}.

\begin{lemma}[{\rm Lemma~3.2 of~\cite{V13}}] \label{lem:to_fractional}
Consider a function $f\colon 2^\cN \to \bR^+$ invariant under a group of permutations $\cG$ on the ground set $\cN$. Let $F$ be the multilinear extension of $f$, $\bar{x} = E_{\sigma \in \cG}[\sigma(x)]$ and fix any $\eps > 0$. Then there is $\delta > 0$ and functions $\hat{F}, \hat{G}\colon [0, 1]^\cN \to \bR^+$ (which are also symmetric with respect to $\cG$), satisfying:
\begin{\itemizeEnv}
	\item For all $x \in [0, 1]^\cN$, $\hat{G}(x) = \hat{F}(\bar{x})$.
	\item For all $x \in [0, 1]^\cN$, $|\hat{F}(x) - F(x)| \leq \eps$.
	\item Whenever $\|x - \bar{x}\|^2_2 \leq \delta$, $\hat{F}(x) = \hat{G}(x)$ and the value depends only on $\bar{x}$.
	\item The first partial derivatives of $\hat{F}, \hat{G}$ are absolutely continuous.
	\item If $f$ is monotone, then $\frac{\partial \hat{F}}{\partial x_i} \geq 0$ and $\frac{\partial \hat{G}}{\partial x_i} \geq 0$ everywhere.
	\item If $f$ is submodular, then $\frac{\partial^2 \hat{F}}{\partial x_i \partial x_j} \leq 0$ and $\frac{\partial^2 \hat{G}}{\partial x_i \partial x_j} \leq 0$ almost everywhere.
\end{\itemizeEnv}
\end{lemma}

Observe that the function $f$ depends only on the size of its input set, and thus, is invariant under any permutation of $\cN$. In other words, in our context: $\bar{x}$ is a vector having the value $|x|/k$ in all the coordinates. Let $\hat{F}$ and $\hat{G}$ be the function guaranteed by Lemma~\ref{lem:to_fractional} when it is applied to $f$ (with the group of all permutations).

\begin{lemma} \label{lem:F_G_opt}
There exists a set of $k$ integral vectors $o_1, o_2, \ldots, o_k \in \{0,1\}^\cN$ such that $\sum_{i = 1}^k o_i = \characteristic_\cN$ and $\sum_{i = 1}^k \hat{F}(o_i) \geq k(1 - \eps)$. On the other hand, for every set of $k$ vectors $x_1, x_2, \ldots, x_k \in [0,1]^\cN$ obeying $\sum_{i = 1}^k x_i = \characteristic_\cN$, it must hold that $\sum_{i = 1}^k \hat{G}(x_i) \leq k[1 - (1 - 1/k)^{k - 1} + \eps]$ and $\sum_{i = 1}^k \hat{F}(x_i) \leq k(1 + \eps)$.
\end{lemma}
\begin{proof}
Recall that $\cN$ contains exactly $k$ elements, and let us name them $v_1, v_2, \ldots,\inJournal{\allowbreak} v_k$ (in an arbitrary order). Let $o_i = \characteristic_{v_i}$, then, clearly, $\sum_{i = 1}^k o_i = \characteristic_\cN$. On the other hand:
\[
	\sum_{i = 1}^k \hat{F}(o_i)
	\geq
	\sum_{i = 1}^k F(o_i) - k\eps
	=
	\sum_{i = 1}^k f(\{v_i\}) - k\eps
	=
	k(1 - \eps)
	\enspace.
\]

Consider now an arbitrary set of $k$ vectors $x_1, x_2, \ldots, x_k \in [0,1]^\cN$ such that: $\sum_{i = 1}^k x_i = \characteristic_\cN$. First observe that for every vector $x \in [0, 1]^\cN$:
\begin{align*}
	F(\bar{x})
	={} &
	\bE\left[1 - \frac{|\RSet(\bar{x})| - 1}{k - 1}\right] - \Pr[\RSet(\bar{x}) = \varnothing] \cdot \left[1 - \frac{0 - 1}{k - 1}\right]\\
	={} &
	1 - \frac{\bE[|\RSet(\bar{x})|] - 1}{k - 1} - \left(1 - \frac{|x|}{k}\right)^k \cdot \frac{k}{k - 1}
	=
	1 - \frac{|x| - 1}{k - 1} - \left(1 - \frac{|x|}{k}\right)^k \cdot \frac{k}{k - 1}
	\enspace.
\end{align*}
Thus:
{\inJournal{\allowdisplaybreaks}
\begin{align*}
	\sum_{i = 1}^k \hat{G}(x_i)
	={} &
	\sum_{i = 1}^k \hat{F}(\bar{x}_i)
	\leq
	\sum_{i = 1}^k F(\bar{x}_i) + k\eps\inJournal{\\}
	=\inJournal{{} &}
	k - \frac{\sum_{i = 1}^k |x_i| - k}{k - 1} - \frac{k}{k - 1} \cdot \sum_{i = 1}^k \left(1 - \frac{|x_i|}{k}\right)^k  + k\eps\\
	\leq{} &
	k - \frac{\sum_{i = 1}^k |x_i| - k}{k - 1} - \frac{k}{k - 1} \cdot \sum_{i = 1}^k \left(1 - \frac{\sum_{i = 1}^k |x_i|}{k^2}\right)^k  + k\eps\\
	={} &
	k - \frac{k^2}{k - 1} \left(1 - \frac{1}{k}\right)^k  + k\eps
	=
	k[1 - (1 - 1/k)^{k - 1} + \eps]
	\enspace.
\end{align*}
}

Finally,
\[
	\sum_{i = 1}^k \hat{F}(x_i)
	\leq
	\sum_{i = 1}^k F(x_i) + k\eps
	\leq
	k(1 + \eps)
	\enspace.
	\qedhere
\]
\end{proof}

Given an arbitrary $n \geq 1$, we construct two instances of {\SW} as follows. Both instance share a single ground set $\cN_n = \cN \times [n]$, and have $k$ players. The utility function (of all the players) in the first and second instances are:
\[
    \hat{f}(S) = \hat{F}\left(\frac{1}{n} \cdot \psi(S)\right)
		\qquad
		\text{and}
		\qquad
		\hat{g}(S) = \hat{G}\left(\frac{1}{n} \cdot \psi(S)\right)
		\enspace,
\]
respectively, where $\psi(S) \in [0, 1]^\cN$ is a vector whose $v$ coordinate counts the number of pairs in $S$ containing $v$. More formally, for every $v \in \cN$,
\[
	\psi_v(S)
	=
	|S \cap (\{v\} \times [n])|
	\enspace.
\]

The following lemma of~\cite{V13} shows (together with the gurantees of Lemma~\ref{lem:to_fractional}) that both $\hat{f}$ and $\hat{g}$ are submodular.
\begin{lemma}[{\rm Lemma~3.1 of \cite{V13}}] \label{lem:implies_submodular}
Let $F\colon [0, 1]^\cN \to \bR$, $n \geq 1$, and define $f\colon 2^{\cN \times [n]} \to \bR$ so that $f(S) = F(x)$ where $x_i = \frac{1}{n} |S \cap (\{i\} \times [n])|$. Then:
\begin{\itemizeEnv}
\item If $\frac{\partial F}{\partial x_i} \geq 0$ everywhere for each $i$, then $f$ is monotone.
\item If the first partial derivatives of $F$ are absolutely continuous and $\frac{\partial^2 F}{\partial x_i \partial x_j} \leq 0$ almost everywhere for all $i, j$, then $f$ is submodular.
\end{\itemizeEnv}
\end{lemma}

The following lemma uses Lemma~\ref{lem:F_G_opt} to bound the values of the optimal solutions of the {\SW} instances corresponding to $\hat{f}$ and $\hat{g}$.

\begin{lemma} \label{lem:f_g_opt_bounds}
Let $\opt_{\hat{f}}$ and $\opt_{\hat{g}}$ denote the optimal values of the two {\SW} instances having $\cN_n$ as the set of items and $k$ players whose (common) objective functions are $\hat{f}$ and $\hat{g}$, respectively. Then: $k(1 - \eps) \leq \opt_{\hat{f}} \leq k(1 + \eps)$ and $\opt_{\hat{g}} \leq k[1 - (1 - 1/k)^{k - 1} + \eps]$.
\end{lemma}
\begin{proof}
By Lemma~\ref{lem:F_G_opt} there exists a set of $k$ integral vectors $o_1, o_2, \ldots, o_k \in \{0,1\}^\cN$ such that $\sum_{i = 1}^k o_i = \characteristic_\cN$ and $\sum_{i = 1}^k \hat{F}(o_i) \geq k(1 - \eps)$. Since each vector $o_i$ is integral, there exists a corresponding set $S_i$ for which $o_i = \characteristic_{S_i}$. Since $\sum_{i = 1}^k o_i = \characteristic_\cN$, the sets $S_1, S_2, \ldots, S_k$ form a partition of $\cN$.

Define $A_i = S_i \times [n]$. Clearly the sets $A_1, A_2, \ldots, A_n$ form a partition of $\cN_n$, and 
\[
	\frac{1}{n} \cdot \psi(A_i) = o_i
	\enspace.
\]
Hence,
\[
	\opt_{\hat{f}}
	\geq
	\sum_{i = 1}^k \hat{f}(A_i)
	=
	\sum_{i = 1}^k \hat{F}(o_i)
	\geq
	k(1 - \eps)
	\enspace.
\]

Next, fix arbitrary $k$ sets $B_1, B_2, \ldots, B_k$ partitioning $\cN_n$. For every element $v \in \cN$,
\[
	\sum_{i = 1}^k \frac{\psi_v(B_i)}{n}
	=
	\sum_{i = 1}^k \frac{|B_i \cap (\{v\} \times [n])|}{n}
	=
	\frac{|\{v\} \times [n]|}{n}
	=
	1
	\enspace.
\]
Hence, $\sum_{i = 1}^k \frac{\psi(B_i)}{n} = \characteristic_\cN$. The last equality implies, by Lemma~\ref{lem:F_G_opt},
\[
	\sum_{i = 1}^k \hat{f}(B_i)
	=
	\sum_{i = 1}^k \hat{F}\left(\frac{\psi(B_i)}{n}\right)
	\leq
	k(1 + \eps)
	\enspace,
\]
and
\[
	\sum_{i = 1}^k \hat{g}(B_i)
	=
	\sum_{i = 1}^k \hat{G}\left(\frac{\psi(B_i)}{n}\right)
	\leq
	k[1 - (1 - 1/k)^{k - 1} + \eps]
	\enspace.
	\qedhere
\]
\end{proof}

The following lemma shows that it is difficult to distinguish between the two above instances of {\SW} (in some sense).

\begin{lemma} \label{lem:no_distinguish}
Assume a uniformly random renaming is applied to the ground set (\ie, every element of $\cN_n$ is unified with a unique uniformly random element from $[nk]$), then any deterministic algorithm distinguishing between $\hat{f}$ and $\hat{g}$ with a constant probability requires an exponential number of value oracle queries.
\end{lemma}
\begin{proof}
Consider a deterministic algorithm $ALG$, and let $D_1, \dotsc, D_h$ be the list of sets whose value is queried by $ALG$ when it is given $\hat{g}$ as the input. Observe that $\hat{g}$ depends on nothing except for the size of its input set, hence, the sets $D_1, \dotsc, D_h$ are identical regardless of the random renaming applied.

Assume, w.l.o.g., that $ALG$ returns one of the sets whose values it queries, and consider what happens when $ALG$ gets $\hat{f}$ as its input. If $\hat{f}(D_i) = \hat{g}(D_i)$ for every set $1 \leq i \leq h$, then $ALG$ is guaranteed to follow the same computation path as when it gets $\hat{g}$, and outputs a set of the same value in both cases. Hence, we only need to show that if $h$ is sub-exponential then with high probability $g(D_i) = f(D_i)$ for every $1 \leq i \leq h$.

Fix some $1 \leq i \leq h$. Let $x = \frac{1}{n}\psi(D_i)$. By definition,
\[
	\bar{x} = \frac{|D_i|}{nk} \cdot \characteristic_{\cN}
	\enspace.
\]
Due to the random renaming, $D_i$ is in fact a random subset of size $|D_i|$ of $\cN_n$. For every $v \in \cN$, $n \cdot x_v$ has a hypergeometric distribution. We bound the probability $n \cdot x_v$ deviates from its expectation using bounds given in \cite{S13} (these bounds are based on results of \cite{C79,H64}). First, observe that $\bE[n \cdot x_v] = |D_i|/k$. Hence,
\begin{align*}
	\Pr\left[\left|x_v - \frac{|D_i|}{nk}\right| > \sqrt{\frac{\delta}{k}}\right]
	=\inJournal{{} &}
	\Pr\left[\left|n \cdot x_v - \bE[n \cdot x_u]\right| > |D_i| \cdot \sqrt{\frac{n^2\delta}{k|D_i|^2}}\right] \inJournal{\\}
	\leq\inJournal{{} &}
	2e^{-2 \cdot \frac{n^2\delta}{k|D_i|^2} \cdot |D_i|}
	\leq
	2e^{-2 \cdot \frac{n\delta}{k^2}}
	\enspace.
\end{align*}

By the union bound, with probability at least $1 - 2ke^{-2 \cdot \frac{n\delta}{k^2}}$, $|x_v - |D_i| / (nk)| \leq \sqrt{\delta/k}$ for every $v \in \cN$, and thus:
\[
	\|x - \bar{x}\|^2_2
	\leq
	k \cdot \left(\sqrt{\frac{\delta}{k}}\right)^2
	=
	\delta
	\enspace.
\]

Hence, by Lemma~\ref{lem:to_fractional}, with probability at least $1 - 2ke^{-2 \cdot \frac{n\delta}{k^2}}$,
\[
	\hat{g}(D_i)
	=
	\hat{G}(x)
	=
	\hat{F}(x)
	=
	\hat{f}(D_i)
	\enspace.
\]

Using the union bound again, we get that with probability $1 - 2khe^{-2 \cdot \frac{n\delta}{k^2}}$, $\hat{f}(D_i) = \hat{g}(D_i)$ for every $1 \leq i \leq h$. The lemma now follows since $\delta$ and $k$ are constants and $h$ is sub-exponential in $n$.
\end{proof}

We are now ready to prove the negative part of Theorem~\ref{thm:identical_submodular_welfare}.

\begin{proof}[\inArXiv{Proof }of the Negative Part of Theorem~\ref{thm:identical_submodular_welfare}]
Fix an arbitrary deterministic algorithm $ALG$ for {\SW} with identical utility functions making a sub-exponential number of value oracle queries. By Lemma~\ref{lem:no_distinguish}, there exists a distribution of instances $\cD$ (produced via the random renaming) such that:
\begin{\itemizeEnv}
	\item Given an instance drawn from $\cD$, $ALG$ finds with probability at least $1 - \eps$ a solution of value at most $\opt_{\hat{g}}$.
	\item No instance in $\cD$ has a solution of value more than $\opt_{\hat{f}}$.
\end{\itemizeEnv}
Hence, given an instance drawn from $\cD$, the expected value of $ALG$'s solution is at most:
\[
	(1 - \eps) \cdot \opt_{\hat{g}} + \eps \cdot \opt_{\hat{f}}
	\leq
	k[1 - (1 - 1/k)^{k - 1} + \eps] + \eps k(1 + \eps)
	\leq
	k[1 - (1 - 1/k)^{k - 1} + \eps] + 2\eps k
	\enspace,
\]
where the first inequality follows from Lemma~\ref{lem:f_g_opt_bounds} and the second one follows by assuming $\eps \leq 1$ (notice that we may assume $\eps$ is smaller than any arbitrary positive constant since proving the theorem for a small value of $\eps$ proves it also for larger values of $\eps$). The approximation ratio of $ALG$ is, therefore, no better than:
\begin{align*}
	\frac{k[1 - (1 - 1/k)^{k - 1} + \eps] + 2\eps k}{\opt_{\hat{f}}}
	\leq{} &
	\frac{k[1 - (1 - 1/k)^{k - 1} + \eps] + 2\eps k}{k(1 - \eps)}\\
	={} &
	\frac{1 - (1 - 1/k)^{k - 1} + 3\eps}{1 - \eps}
	\leq
	1 - (1 - 1/k)^{k - 1} + 6\eps
	\enspace,
\end{align*}
where the last inequality assumes $\eps \leq 1/3$. This completes the proof of the theorem for deterministic algorithms. The proof extends to randomized algorithms via Yao's Principle since we have found a single distribution $\cD$ which is difficult for every deterministic algorithm using a sub-exponential number of value oracle queries.
\end{proof}

\inArXiv{
\inArXiv{\paragraph{Acknowledgment.}}
\inJournal{\begin{acknowledgment}}
We would like to thank Uri Feige for pointing out the relevance of the result of \citeWithName{Khot et al.}{KLMM08} for our work, and for simplifying the proof of Theorem~\ref{thm:identical_submodular_welfare}.
\inJournal{\end{acknowledgment}}
\bibliographystyle{plain}
\bibliography{submodular}
}

\appendix
\inJournal{\section*{APPENDIX}}
\section{Hardness of Cardinality Constraints under Symmetric Objectives} \label{sec:hardness}

In this section we prove Theorem~\ref{thm:hardness_uniform_symmetric}.

\begin{reptheorem}{thm:hardness_uniform_symmetric}
Consider the problems $\max\{f(S) : |S| = p/q \cdot n\}$ and $\max\{f(S) : |S| \leq p/q \cdot n\}$ where $p < q$ are positive constant integers and $f$ is a non-negative symmetric submodular function $f\colon 2^\cN \to \bR^+$ obeying $n / q \in \bZ$. Then, every algorithm with an approximation ratio of $\nicefrac[]{1}{2} + \eps$ for one of the above problems (for any constant $\eps > 0$) uses an exponential number of value oracle queries.
\end{reptheorem}

The classes of problems referred to by Theorem~\ref{thm:hardness_uniform_symmetric} are closed under the refinement defined by Definition~1.7 of~\cite{V13} (for given $p$ and $q$). Thus, by Theorem~1.8 of~\cite{V13}, to prove Theorem~\ref{thm:hardness_uniform_symmetric} it is enough to find (for given $p$ and $q$) a symmetric submodular function $f_{p, q}\colon 2^\cN \to \bR^+$ (and a corresponding multilinear extension $F_{p, q}$) obeying:
\begin{\itemizeEnv}
	\item $\cN = \{1, 2, \dotsc, 2q\}$.
	\item There exists a set $S' \subseteq \cN$ of size $2p$ such that $f_{p, q}(S') = 1$.
	\item There exists a permutation $\sigma\colon \cN \to \cN$ such that: $f_{p, q}(S) = f_{p, q}(\sigma(S))$ for every set $S \subseteq \cN$ and $F_{p, q}(x) \leq 1/2$ for every vector $x \in \{z \in [0, 1]^\cN : \sigma(z) = z\}$.
\end{\itemizeEnv}

\begin{proof}[\inArXiv{Proof }of Theorem~\ref{thm:hardness_uniform_symmetric}]
We define $f_{p, q}$ as follows:
\[
	f_{p, q}(S) =
	\begin{cases}
		1 & \text{if $|\{1, 2q\} \cap S| = 1$} \enspace, \\
		0 & \text{otherwise} \enspace.
	\end{cases}
\]

Let us show that $f_{p, q}$ has all the required properties. First observe that $f_{p, q}$ can be presented as the cut function of a graph containing $2q$ nodes and a single edge, hence, it is symmetric and submodular. Also, the set $S' = \{1, 2, \dotsc, 2p\}$ is of size $2p$ and have $f_{p, q}(S') = 1$.

Consider now the permutation $\sigma$ mapping every node $i$ to $2q - i + 1$. Since this permutation maps $1$ and $2q$ to each other, we get $f_{p, q}(S) = f_{p, q}(\sigma(S))$ for every set $S \subseteq \cN$. Moreover, every vector $x \in \{z \in [0, 1]^\cN : \sigma(z) = z\}$ must have: $z_1 = z_{2q}$. Hence,
\[
	F_{p, q}(x)
	=
	z_1(1 - z_{2q}) + z_{2q}(1 - z_1)
	=
	2z_1(1 - z_1)
	\leq
	\frac{1}{2}
	\enspace.
	\qedhere
\]
\end{proof}
\section{Proof of Theorem~\ref{thm:uniform_base_general_approximation}} \label{sec:uniform_base_approximation_general}

The algorithm we use to prove Theorem~\ref{thm:uniform_base_general_approximation} is Algorithm~\ref{alg:DoubleContinuousGreedy2}, which is a close variant of Algorithm~\ref{alg:DoubleContinuousGreedy}. The two algorithms defer in three points:
\begin{\itemizeEnv}
	\item The way $T$ is set.
	\item The method of choosing $I^1(t)$ (and therefore, also $I^2(t)$).
	\item The third ``foreach'' loop is removed.
\end{\itemizeEnv}
We observe that all the proofs of Section~\ref{sec:cardinality} up to Corollary~\ref{cor:value_as_worse_option} can be made to work with these changes. In other words, we know that Algorithm~\ref{alg:DoubleContinuousGreedy2} is a polynomial time algorithm whose output $y$ is a feasible solution obeying $F(y) \geq \min\{F(y^1(T)), F(y^2(T))\}$.

\begin{algorithm}[!ht]
\caption{\textsf{Double Measured Continuous Greedy - General Submodular Objectives}$(f, \cN, k)$} \label{alg:DoubleContinuousGreedy2}
\DontPrintSemicolon
\tcp{Initialization}
Set: $T \gets 1$ and $\delta \gets T (\lceil n^5 T \rceil)^{-1}$.\\
Initialize: $t \gets 0$, $y^1(0) \gets \characteristic_\varnothing$ and $y^2(0) \gets \characteristic_\cN$.\\

\BlankLine

\tcp{Main loop}
\While{$t < T$}
{
    \ForEach{$u \in \cN$}
    {
        Let $w^1_u(t) \gets F(y^1(t) \vee \characteristic_u) - F(y^1(t))$ and $w^2_u(t) \gets F(y^2(t) \wedge \characteristic_{\cN - u}) - F(y^2(t))$.
    }
    Let $I^1(t) \in [0, 1]^\cN$ and $I^2(t) \in [0, 1]^\cN$ be two vectors maximizing
		\[
			\min\{I^1(t) \cdot w^1(t) + F(y^1(t)), I^2(t) \cdot w^2(t) + F(y^2(t)) \}
		\]
		among the vectors obeying $|I^1(t)| = k$, $|I^2(t)| = n - k$ and $I^1(t) + I^2(t) = \characteristic_\cN$.\\
    \ForEach{$u \in \cN$}
    {
        Let $y^1_u(t + \delta) \gets y^1_u(t) + \delta I^1_u(t) \cdot (1 - y^1_u(t))$ and $y^2_u(t + \delta) \gets y^2_u(t) - \delta I^2_u(t) \cdot y^2_u(t)$.
    }
    $t \leftarrow t + \delta$.
}

\BlankLine

\lIf{$|y^1(T)| = |y^2(T)|$}
{
	\Return{$y^1(T)$}.
}
\lElse
{
	\Return{$y^1(T) \cdot \frac{|y^2(T)| - k}{|y^2(T)| - |y^1(T)|} + y^2(T) \cdot \frac{k - |y^1(T)|}{|y^2(T)| - |y^1(T)|}$}.
}
\end{algorithm}

The following lemma and corollary give a lower bound on the improvement of the solutions maintained by Algorithm~\ref{alg:DoubleContinuousGreedy2} in every iteration. These lemma and corollary are the counter-part of Lemma~\ref{lem:step_approximation_two_sided} and Corollary~\ref{cor:step_improvment_two_sided} from Section~\ref{sec:cardinality}. Let $\Delta(t) = \min\{F(y^1(t) \vee \characteristic_{OPT}), F(y^2(t) \wedge \characteristic_{OPT})\}$

\begin{lemma} \label{lem:step_approximation_two_sided_general}
For every time $0 \leq t < T$:
\[
	\sum_{u \in \cN} (1 - y^1_u(t)) \cdot I^1_u(t) \cdot \partial_u F(y^1(t)) + F(y^1(t)) \geq \Delta(t)
	\enspace,
\]
\[
	-\sum_{u \in \cN} y^2_u(t) \cdot I^2_u(t) \cdot \partial_u F(y^2(t)) + F(y^2(t)) \geq \Delta(t)
	\enspace.
\]
\end{lemma}
\begin{proof}
Let us calculate the weights of $OPT$ and $\cOPT$ according to the weight functions $w^1(t)$ and $w^2(t)$, respectively.
\[
    w^1(t) \cdot \characteristic_{OPT}
    =
    \sum_{u \in OPT} w^1_u(t)
    =
    \sum_{u \in OPT} \left[F(y^1(t) \vee \characteristic_u) - F(y^1(t))\right]
		\geq
    F(y^1(t) \vee \characteristic_{OPT}) - F(y^1(t)) \enspace,
\]
and
\[
    w^2(t) \cdot \characteristic_{\cOPT}
    =
    \sum_{u \in \cOPT} w^2_u(t)
    =
    \sum_{u \in \cOPT} \left[F(y^2(t) \wedge \characteristic_{\cN - u}) - F(y^2(t))\right]
		\geq
    F(y^2(t) \wedge \characteristic_{OPT}) - F(y^2(t)) \enspace,
\]
where the inequalities follow from submodularity. Since $|OPT| = k$,
\begin{align*}
	&
	\min\{I^1(t) \cdot w^1(t) + F(y^1(t)), I^t(t) \cdot w^2(t) + F(y^2(t)) \}\\
	\geq{} &
	\min\{\characteristic_{OPT} \cdot w^1(t) + F(y^1(t)), \characteristic_{\cOPT} \cdot w^2(t) + F(y^2(t)) \}\\
	\geq{} &
	\min\{F(y^1(t) \vee \characteristic_{OPT}), F(y^2(t) \wedge \characteristic_{OPT})\}
	=
	\Delta(t)
	\enspace.
\end{align*}
Hence,
\begin{align*}
    \sum_{u \in \cN} (1 - y^1_u(t)) \cdot I^1_u(t) \cdot \partial_u F(y^1(t))
    ={} &
    \sum_{u \in \cN} I^1_u(t) \cdot [F(y^1(t) \vee \characteristic_u) - F(y^1(t))]
    =
    I^1(t) \cdot w^1(t)\\
    \geq{} &
    \Delta(t) - F(y^1(t)) \enspace,
\end{align*}
and
\begin{align*}
    -\sum_{u \in \cN} y^2_u(t) \cdot I^2_u(t) \cdot \partial_u F(y^2(t))
    ={} &
    \sum_{u \in \cN} I^2_u(t) \cdot [F(y^2(t) \wedge \characteristic_{\cN - u}) - F(y^2(t))]
    =
    I^2(t) \cdot w^2(t)\\
    \geq{} &
    \Delta(t) - F(y^2(t)) \enspace.
		\qedhere
\end{align*}
\end{proof}

\begin{corollary} \label{cor:step_improvment_two_sided_general}
For every time $0 \leq t < T$,
\[
	F(y^1(t + \delta)) - F(y^1(t)) \geq \delta \cdot [\Delta(t) - F(y^1(t))] - O(n^3\delta^2) \cdot f(OPT)
	\enspace,
\]
and
\[
	F(y^2(t + \delta)) - F(y^2(t)) \geq \delta \cdot [\Delta(t) - F(y^2(t))] - O(n^3\delta^2) \cdot f(OPT)
	\enspace,
\]
\end{corollary}
\begin{proof}
The corollary follows immediately from Lemmata~\ref{lem:greedy_step_increase_bound_general}, \ref{le:OPT_not_too_bad}\footnote{The proof of Lemma~\ref{le:OPT_not_too_bad} does not use the symmetry of $f$.} and~\ref{lem:step_approximation_two_sided_general}.
\end{proof}

To use the lower bounds give by Corollary~\ref{cor:step_improvment_two_sided_general}, we need the following lemma of~\cite{FNS11}.

\begin{lemma}[{\rm Lemma~III.5 of \cite{FNS11}}] \label{lem:target_lower_bound}
Consider a vector $x \in [0, 1]^\cN$. Assuming $x_u \leq a$ for every $u \in \cN$, then for every set $S \subseteq \cN$, $F(x \vee \characteristic_{S}) \geq (1 - a)f(S)$.
\end{lemma}

We notice that Lemma~\ref{lem:target_lower_bound} applies also to $\bar{F}$ since $\bar{f}$ is also submodular. The following lemma is a counterpart of Lemma~III.6 of \cite{FNS11}.

\begin{lemma} \label{lem:max_y}
For every time $0 \leq t \leq T$ and element $u \in \cN$, $\max\{y^1_u(t), 1 - y^2_u(t)\} \leq 1 - (1 - \delta)^{t / \delta} \leq 1 - e^{-t} + O(\delta)$.
\end{lemma}
\begin{proof}
We first prove the inequality $y^1_u(t) \leq 1 - (1 - \delta)^{t / \delta}$. The proof is done by induction on $t$. For $t = 0$, the inequality holds because $y^1_u(0) = 0 = 1 - (1 - \delta)^{0/\delta}$. Assume the inequality holds for some $t$, and let us prove it for $t + \delta$.
\begin{align*}
    y^1_u(t + \delta)
    ={} &
    y^1_u(t) + \delta I^1_u(t) (1 - y^1_u(t))
    =
    y^1_u(t) (1 - \delta I^1_u(t)) + \delta I^1_u(t)\\
    \leq{} &
    (1 - (1 - \delta)^{t / \delta}) (1 - \delta I^1_u(t)) + \delta I^1_u(t)
    =
    1 - (1 - \delta)^{t / \delta} + \delta I^1_u(t)(1 - \delta)^{t / \delta}\\
    \leq{} &
    1 - (1 - \delta)^{t / \delta} + \delta (1 - \delta)^{t / \delta}
    =
    1 - (1 - \delta)^{(t + \delta) / \delta} \enspace.
\end{align*}
The proof that $1 - y^2_u(t) \leq 1 - (1 - \delta)^{t / \delta}$ is analogous to the above proof. To complete the proof of the lemma, we still need show that the inequality $1 - (1 - \delta)^{t / \delta} \leq 1 - e^{-t} + O(\delta)$ holds:
\[
    1 - (1 - \delta)^{t / \delta}
    \leq
    1 - [e^{-1}(1 - \delta)]^t
    =
    1 - e^{-t}(1 - \delta)^t
    \leq
    1 - e^{-t}(1 - t \delta)
    \leq
    1 - e^{-t} + O(\delta)
		\enspace,
\]
where the last inequality holds since $e^{-t}t \leq e^{-1}$ for every $t$.
\end{proof}

\begin{corollary} \label{cor:delta_bound_general}
For every time $0 \leq t < T$, $\Delta(t) \geq (e^{-t} - O(\delta)) \cdot f(OPT)$
\end{corollary}
\begin{proof}
By Lemmata~\ref{lem:target_lower_bound} and \ref{lem:max_y},
\[
	F(y^1(t) \vee \characteristic_{OPT})
	\geq
	(e^{-t} - O(\delta)) \cdot f(OPT)
	\enspace,
\]
and
\[
	F(y^2(t) \wedge \characteristic_{OPT})
	=
	\bar{F}((\characteristic_\cN - y^2(t)) \vee \characteristic_{\cOPT})
	\geq
	(e^{-t} - O(\delta)) \cdot \bar{f}(\cOPT)
	=
	(e^{-t} - O(\delta)) \cdot f(OPT)
	\enspace.
\]
The lemma follows by plugging both observations into the definition of $\Delta(t)$.
\end{proof}

Combining Corollaries~\ref{cor:step_improvment_two_sided_general} and~\ref{cor:delta_bound_general}, we get the following corollary.

\begin{corollary} \label{cor:step_improvment_two_sided_assignment}
For every time $0 \leq t < T$,
\[
	F(y^1(t + \delta)) - F(y^1(t)) \geq \delta \cdot [e^{-t} \cdot f(OPT) - F(y^1(t))] - O(n^3\delta^2) \cdot f(OPT)
	\enspace,
\]
and
\[
	F(y^2(t + \delta)) - F(y^2(t)) \geq \delta \cdot [e^{-t} \cdot f(OPT) - F(y^2(t))] - O(n^3\delta^2) \cdot f(OPT)
	\enspace,
\]
\end{corollary}

In order to complete the analysis of Algorithm~\ref{alg:DoubleContinuousGreedy2}, we need to derive from the last corollary lower bounds on $F(y^1(T))$ and $F(y^2(T))$. This derivation is identical to the one used by~\cite{FNS11} to derive their result from their Corollary~III.7 (which is the counterpart of Corollary~\ref{cor:step_improvment_two_sided_assignment}). We give the proof again below for completeness.

Let $g(t)$ be defined as follows: $g(0) = 0$ and $g(t + \delta) = g(t) + \delta [e^{-t} \cdot f(OPT) - g(t)]$. The next lemma shows that a lower bound on $g(t)$ also gives a lower bound on $F(y^1(t))$ and $F(y^2(t))$.

\begin{lemma} \label{lem:g_bound}
For every time $0 \leq t \leq T$, $g(t) \leq \min\{F(y^1(t)), F(y^2(t))\} + O(n^3\delta) \cdot t \cdot f(OPT)$.
\end{lemma}
\begin{proof}
We prove $g(t) \leq F(y^1(t)) + O(n^3\delta^{-2}) \cdot t \cdot f(OPT)$. The proof of the claim for $F(y^2(t))$ is analogous. Let $c$ be the constant hiding behind the big $O$ notation in Corollary~\ref{cor:step_improvment_two_sided_assignment}. We prove by induction on $t$ that $g(t) \leq F(y^1(t)) + c n^3\delta t \cdot f(OPT)$. For $t = 0$, $g(0) = 0 \leq F(y^1(0))$. Assume now that the claim holds for some $t$, and let us prove it for $t + \delta$. Using Corollary~\ref{cor:step_improvment_two_sided_assignment}, we get:
\begin{align*}
    g(t + \delta)
    &=
    g(t) + \delta [e^{-t} \cdot f(OPT) - g(t)]
    =
    (1 - \delta) g(t) + \delta e^{-t} \cdot f(OPT)\\
    &\leq
    (1 - \delta) [F(y^1(t)) + cn^3\delta t \cdot f(OPT)] + \delta e^{-t} \cdot f(OPT)\\
    &=
    F(y^1(t)) + \delta [e^{-t} \cdot f(OPT) - F(y^1(t))] + c(1 - \delta) n^3\delta t \cdot f(OPT)\\
    &\leq
    F(y^1(t + \delta)) + cn^3\delta^2 \cdot f(OPT) + c(1 - \delta)n^3\delta t \cdot f(OPT)\\
    &\leq
    F(y^1(t + \delta)) + cn^3\delta(t + \delta) \cdot f(OPT)
		\enspace.
		\qedhere
\end{align*}
\end{proof}

The function $g$ is given by a recursive formula, thus, evaluating it is not immediate. Instead, we show that the function $h(t) = t e^{-t} \cdot f(OPT)$ lower bounds $g$ within the range $[0, 1]$ (recall that Algorithm~\ref{alg:DoubleContinuousGreedy2} sets $T = 1$).
\begin{lemma} \label{lem:h_bound}
For every $0 \leq t \leq 1$, $g(t) \geq h(t)$.
\end{lemma}
\begin{proof}
The proof is by induction on $t$. For $t = 0$, $g(0) = 0 = 0 \cdot e^{-0} \cdot f(OPT) = h(0)$. Assume now that the lemma holds for some $t$, and let us prove it holds for $t + \delta$.
\begin{align*}
    h(t + \delta)
    ={} &
    h(t) + \int_t^{t + \delta} h'(\tau) d\tau
    =
    h(t) + f(OPT) \cdot \int_t^{t + \delta} e^{-\tau}(1 - \tau) d\tau\\
    \leq{} &
    h(t) + f(OPT) \cdot \delta e^{-t}(1 - t)
    =
    (1 - \delta) h(t) + \delta e^{-t} \cdot f(OPT)\\
    \leq{} &
    (1 - \delta) g(t) + \delta e^{-t} \cdot f(OPT)
    =
    g(t) + \delta \cdot [e^{-t} \cdot f(OPT) - g(t)]
    =
    g(t + \delta) \enspace. \qedhere
\end{align*}
\end{proof}

The last result implies lower bounds on $F(y^1(T))$ and $F(y^2(T)$.
\begin{corollary} \label{cor:final_value_no_enum}
$\min\{F(y^1(T)), F(y^2(T))\} \geq [e^{-1} - o(1)] \cdot f(OPT)$.
\end{corollary}
\begin{proof}
By Lemmata~\ref{lem:g_bound} and \ref{lem:h_bound}:
\begin{align*}
    \min\{F(y^1(T)), F(y^2(T))\}
		={} &
		\min\{F(y^1(1)), F(y^2(1))\}
    \geq
    g(1) - O(n^3\delta) \cdot 1 \cdot f(OPT)\\
    \geq{} &
    h(1) - O(n^{-2}) \cdot f(OPT)
    =
    [e^{-1} - O(n^{-2})] \cdot f(OPT)
		\enspace.
		\qedhere
\end{align*}
\end{proof}

The approximation ratio guaranteed by Theorem~\ref{thm:uniform_base_general_approximation} follows immediately from Corollaries~\ref{cor:value_as_worse_option} and~\ref{cor:final_value_no_enum}.
\section{An Alternative Proof of the Positive Part of Theorem~\ref{thm:identical_submodular_welfare}} \label{app:SWOriginal}

In this section we give an analysis of the approximation ratio of Algorithm~\ref{alg:IdenticalSW} which is different than the one given in Section~\ref{ssc:positive_welfare}. First, we need to define some notation. 
Given a collection $\cA$ of disjoint sets and an integer $0 \leq h \leq |\cA|$, let $\RSet(\cA, h)$ be the random set resulting from taking the union of $h$ sets of $\cA$ chosen uniformly at random (without replacements).

Lemma~\ref{lem:disjoint_unions} relates $\bE[f(\RSet(\cA, h))]$ with the average value of a set of $\cA$.

\begin{lemma} \label{lem:disjoint_unions}
Given a non-negative submodular function $f\colon 2^\cN \to \bR^+$, a collection $\cA$ of $\ell \geq 2$ \emph{disjoint} subsets $A_1, A_2, \ldots, A_\ell$ of the ground set $\cN$, and an integer $1 \leq h \leq \ell$, then:
\[
	\bE[f(\RSet(\cA, h))]
	\geq
	\left(1 - \frac{h - 1}{\ell - 1}\right) \cdot \frac{\sum_{i=1}^\ell f(A_i)}{\ell}
	\enspace.
\]
\end{lemma}
\begin{proof}
Observe that:
\begin{align*}
	\bE[f(\RSet(\cA, h))]
	={} &
	{\binom{\ell}{h}}^{-1} \cdot \sum_{\substack{\cB \subseteq \cA \\ |\cB| = h}} f\mathopen{}\left({\textstyle \bigcup_{A_i \in \cB} A_i} \right)
	=
	h^{-1}{\binom{\ell}{h}}^{-1} \cdot \sum_{i = 1}^\ell \sum_{\substack{\cB \subseteq \cA - A_i \\ |\cB| = h - 1}} f\mathopen{}\left(A_i \cup {\textstyle \bigcup_{A_j \in \cB} A_j} \right)\\
	={} &
	\frac{\sum_{i = 1}^\ell \mathbb{E}[f(A_i \cup \RSet(\cA - A_i, h - 1))]}{\ell}
	\enspace.
\end{align*}
Since the sets of $\cA$ are disjoint, an element appears in $\RSet(\cA - A_i, h - 1)$ with probability at most $(h - 1)/(\ell - 1)$. On the other hand, since $f(A_i \cup S)$ is a non-negative submodular function of $S$, we get by Lemma~\ref{lem:max_probability_max_damage} that:
\[
	\mathbb{E}[A_i \cup f(\RSet(\cA - A_i, h - 1))]
	\geq
	\left(1 - \frac{h - 1}{\ell - 1}\right) \cdot f(A_i \cup \varnothing)
	=
	\left(1 - \frac{h - 1}{\ell - 1}\right) \cdot f(A_i)
	\enspace.
\]
The lemma follows by combining the above results.
\end{proof}

The next lemma bounds the value of a certain random set obtained by taking the union of multiple random sets. Notice that Lemmata~2.2 and~2.3 of~\cite{FMV11} correspond to the cases of $\ell = 1$ and $\ell = 2$ of this lemma. The proof of the lemma is based on the same technique used by \citeWithName{Feige et al.}{FMV11} to derive their Lemma~2.3 from their Lemma~2.2.

\begin{lemma} \label{lem:equal_prob_repeated_bound}
Given a non-negative submodular function $f\colon 2^\cN \to \bR^+$, $\ell$ subsets $A_1, A_2, \ldots, A_\ell$ of the ground set $\cN$ and a probability $p$, then:
\[
	\mathbb{E}\mathopen{}\left[f\mathopen{}\left({\textstyle \bigcup_{i = 1}^\ell A_i(p)}\right)\right]
	\geq
	\sum_{I \subseteq [\ell]} p^{|I|}(1 - p)^{\ell - |I|} \cdot f\mathopen{}\left({\textstyle \bigcup_{i \in I} A_i} \right)
	\enspace,
\]
assuming the random sets $\{A_i(p)\}_{i = 1}^\ell$ are independent and $\bigcup_{A_i \in \varnothing} A_i$ is defined as $\varnothing$.
\end{lemma}
\begin{proof}
The proof is by induction on $\ell$. The case $\ell = 1$ is identical to Lemma~\ref{lem:equal_prob_bound} (Lemma~2.2 of \cite{FMV11}). Assume the lemma holds for $\ell - 1 \geq 1$, and let us prove it for $\ell$. Then,
\begin{align*}
	\mathbb{E}\mathopen{}\left[f\mathopen{}\left({\textstyle \bigcup_{i = 1}^\ell A_i(p)}\right)\right]
	={} &
	\sum_{A' \subseteq A_\ell} \Pr[A_\ell(p) = A'] \cdot \mathbb{E}\mathopen{}\left[f\mathopen{}\left({\textstyle \bigcup_{i = 1}^\ell A_i(p)}\right) \mid A_\ell(p) = A'\right]\\
	={} &
	\sum_{A' \subseteq A_\ell} \Pr[A_\ell(p) = A'] \cdot \mathbb{E}\mathopen{}\left[f\mathopen{}\left(A' \cup {\textstyle \bigcup_{i = 1}^{\ell - 1} A_i(p)}\right)\right]
	\enspace.
\end{align*}

Since $f(A' \cup S)$ is a non-negative submodular function of $S$, we get by the induction hypothesis that:
\[
	\mathbb{E}\mathopen{}\left[\mathopen{}f\left(A' \cup {\textstyle \bigcup_{i = 1}^{\ell - 1} A_i(p)} \right)\right]
	\geq
	\sum_{I \subseteq [\ell-1]} p^{|I|}(1 - p)^{(\ell - 1) - |I|} \cdot f\mathopen{}\left(A' \cup {\textstyle \bigcup_{i \in I} A_i} \right)
	\enspace.
\]

Combining the above equality and inequality, and changing the order of summation, gives:
\begin{align*}
	\mathbb{E}\mathopen{}\left[f\mathopen{}\left({\textstyle \bigcup_{i = 1}^\ell A_i(p)}\right)\right]
	\geq{} &
	\sum_{I \subseteq [\ell-1]} p^{|I|}(1 - p)^{(\ell - 1) - |I|} \cdot \left[ \sum_{A' \subseteq A_\ell} \Pr[A_\ell(p) = A'] \cdot f\mathopen{}\left(A' \cup {\textstyle \bigcup_{i \in I} A_i} \right)\right]\\
	={} &
	\sum_{I \subseteq [\ell-1]} p^{|I|}(1 - p)^{(\ell - 1) - |I|} \cdot \bE \mathopen{}\left[ f\mathopen{}\left(A_\ell(p) \cup {\textstyle \bigcup_{i \in I} A_i} \right)\right]\\
	\geq{} &
	\sum_{I \subseteq [\ell-1]} p^{|I|}(1 - p)^{(\ell - 1) - |I|} \cdot \left[ (1- p) f\mathopen{}\left({\textstyle \bigcup_{i \in I} A_i} \right) + p \cdot f\mathopen{}\left({\textstyle \bigcup_{i \in I + \ell} A_i} \right) \right]
	\enspace,
\end{align*}
where the last inequality follows from Lemma~\ref{lem:equal_prob_bound} since $f\mathopen{}\left(A_\ell(p) \cup {\textstyle \bigcup_{i \in I} A_i} \right)$ is a submodular function of $A_\ell(p)$.
\end{proof}

We are now ready to prove the positive part of Theorem~\ref{thm:identical_submodular_welfare}.

\begin{proof}[\inArXiv{Proof }of the Positive Part of Theorem~\ref{thm:identical_submodular_welfare}]
Let $OPT_i$ denote the set of elements assigned to player $p_i$ by the optimal solution, and let $\cO = \{OPT_i\}_{i = 1}^k$. By the linearity of the expectation, we can bound the value of the solution produced by Algorithm~\ref{alg:IdenticalSW} as follows:
\begin{equation} \label{eq:rephrasing_S}
	\bE\mathopen{}\left[\sum_{i = 1}^k f(S_i)\right]
	=
	k \cdot \bE[f(\cN(1/k))]
	=
	k \cdot \bE\mathopen{}\left[f\mathopen{}\left({\textstyle \bigcup_{i = 1}^k OPT_i(1/k)}\right)\right]
	\enspace.
\end{equation}

By Lemmata~\ref{lem:disjoint_unions} and~\ref{lem:equal_prob_repeated_bound}:
\begin{align*}
	\bE\mathopen{}\left[f\mathopen{}\left({\textstyle \bigcup_{i = 1}^k OPT_i(1/k)}\inJournal{\right.\right.&\left.\left.\mspace{-14mu}\vphantom{\textstyle\bigcup_{i = 1}^k}}\right)\right]
	\geq\inArXiv{{} &}
	\sum_{I \subseteq [k]} (1/k)^{|I|}(1 - 1/k)^{k - |I|} \cdot f\mathopen{}\left({\textstyle \bigcup_{i \in I} OPT_i} \right)\\
	\geq{} &
	\sum_{h = 1}^{k - 1} \left[(1/k)^{h}(1 - 1/k)^{k - h} \cdot \binom{k}{h} \cdot \bE[f(\RSet(\cO, h))] \right]\\
	\geq{} &
	\sum_{h = 1}^{k - 1} \left[(1/k)^{h}(1 - 1/k)^{k - h} \cdot \binom{k}{h} \cdot \left(1 - \frac{h - 1}{k - 1}\right) \cdot \frac{\sum_{i = 1}^k f(OPT_i)}{k} \right]
	\enspace.
\end{align*}

Rearranging the rightmost hand side of the above inequality yields:
\begin{align*}
	\bE\mathopen{}\left[f\mathopen{}\left({\textstyle \bigcup_{i = 1}^k OPT_i(1/k)}\right)\right]
	\geq{} &
	\frac{\sum_{i = 1}^k f(OPT_i)}{k} \cdot \sum_{h = 1}^{k - 1} \left[\binom{k - 1}{h} \cdot (1/k)^{h}(1 - 1/k)^{(k - 1) - h} \right]\\
	={} &
	\frac{\sum_{i = 1}^k f(OPT_i)}{k} \cdot \left[1 - \binom{k - 1}{0} \cdot (1/k)^{0}(1 - 1/k)^{(k - 1) - 0}\right]\\
	={} &
	\frac{\sum_{i = 1}^k f(OPT_i)}{k} \cdot [1 - (1 - 1/k)^{k - 1}]
	\enspace,
\end{align*}
where the first equality holds since $\sum_{h = 0}^{k - 1} \left[\binom{k - 1}{h} \cdot a^{h}b^{(k - 1) - h} \right] = 1$ whenever $a + b = 1$. The theorem now follows by plugging the last inequality into Equality~\eqref{eq:rephrasing_S}.
\end{proof}

\inJournal{
\inArXiv{\paragraph{Acknowledgment.}}
\inJournal{\begin{acknowledgment}}
We would like to thank Uri Feige for pointing out the relevance of the result of \citeWithName{Khot et al.}{KLMM08} for our work, and for simplifying the proof of Theorem~\ref{thm:identical_submodular_welfare}.
\inJournal{\end{acknowledgment}}
\bibliographystyle{ACM-Reference-Format-Journals}
\bibliography{submodular}
\received{Month Year}{Month Year}{Month Year}
}

\end{document}